\documentclass[a4paper,11pt]{article}
\usepackage{jheppub} 
\usepackage{lineno}
\usepackage{amsmath}
\usepackage{amssymb}
\usepackage{amsfonts}
\usepackage{amssymb}
\usepackage{amsthm}
\usepackage{amscd}
\usepackage{xcolor}
\usepackage{appendix}
\usepackage{adjustbox}
\usepackage{tabularray}

\newcommand{\R}{\mathbb{R}}%

\newcommand{\ew}{\mathcal{E}}%
\newcommand{\pew}{\partial\mathcal{E}}%
\newcommand{\bb}{\partial M}%
\newcommand{\HH}{\mathcal{H}}%
\newcommand{\Vset}{{V}_{set}}%

\newcommand{\cS}{\mathcal{S}}

\newcommand{\nset}{\{1,\cdots,n\}}

\newcommand{\HHRT}{\mathrm{HRT}}
\newcommand{\RT}{\mathrm{RT}}

\newcounter{Counter}
\newtheorem{definition}{Definition}[section]
\newtheorem{lemma}[definition]{Lemma}

\newtheorem{theorem}[definition]{Theorem}
\newtheorem{proposition}[definition]{Proposition}
\newtheorem{corollary}[definition]{Corollary}

\newtheorem{remark}[definition]{Remark}
\newtheorem{assumption}{Assumption}

\title{\boldmath Graph models for covariant holographic entropy I}

\author{Bowen Zhao}
\affiliation{Beijing Institute of Mathematical Sciences and Applications, Beijing, China}

\emailAdd{bowenzhao@bimsa.cn}

\abstract{
We construct a graph model for holographic entropies in general time-dependent spacetimes. In static settings, such models arise from Ryu–Takayanagi surfaces on a common Cauchy slice and imply that the holographic entropy cone is polyhedral. Extending this construction to the covariant Hubeny–Rangamani–Takayanagi (HRT) setting is obstructed by the absence of a preferred time slice, raising the possibility of unphysical “short-cuts” built from partial HRT surfaces.

We identify a geometric condition---the existence of exposed regions for each pair of HRT surfaces---under which this obstruction is removed. Under this condition, we construct weight functions by projecting along null generators of entanglement horizons and prove a Conditional No-Short-Cut Theorem: any graph cut is dominated by a surface composed of complete HRT surfaces. Consequently, the graph model reproduces HRT entropies, establishing the equivalence between the covariant and static holographic entropy cones in this regime.

We further show that configurations in which exposed regions are absent due to nesting of interaction regions can be partially resolved by grouping HRT surfaces into timelike clusters. This provides evidence that the graph model extends beyond the exposed-region regime and suggests a path toward a complete covariant construction.

}


\begin{document}
\maketitle
\flushbottom

\section{Introduction}

\subsection{Background: Graph models and the RT cone}
The entanglement entropy of subregions in a quantum system provides a powerful probe of its structure. For states in holographic conformal field theories with semi-classical gravitational duals, the Ryu-Takayanagi (RT) and Hubeny-Rangamani-Takayanagi (HRT) formulae ground this study in geometry \cite{RT2006formula,HRT2007covariant}, famously equating boundary entropy with the area of a bulk extremal surface. The set of all entropy vectors for $n$ disjoint, spacelike-separated boundary regions forms the holographic entropy cone, whose facets correspond to universal entropy inequalities constraining any dual gravitational theory.

A major advance in understanding this structure was achieved by Bao et al. \cite{bao2015holographic}, who showed that for static or time-symmetric bulk geometries the holographic entropy cone is polyhedral. Their proof relies on a graph model constructed from RT surfaces on a fixed Cauchy slices. This construction allows one to apply max-flow/min-cut techniques, reducing entropy inequalities to combinatorial statments about cuts in a weighted graph.

A central open problem is whether an analogous description exists in the fully covariant setting, where entropies are computed by HRT surfaces that are not confied to a common Cauchy slice. Establishing such a description would imply that the covariant holographic entropy cone coincides with the static one, thereby extending the polyhedrality result to general time-dependent states.

Substantial evidence has been accumulated in recent years to support this conjecture \cite{bao2018entropy_large_region_late_time,erdmenger2017HECnumerics,caginalp2020holographicAdSVaidya_numerics,czech2019holographic2+1dimension,grimaldi2025newcharacterHEC}. These studies provides strong indications that the combinatorial structure underlying holographic entropy inequalities is insensitive to time dependence. These results motivate the search for a direct geometric construction of a graph model in the covariant setting.

Beyond its implications for entropy inequalities, a covariant graph model may also shed light on the construction of holographic tensor networks for time-dependent states \cite{bao2019holographicTensorNetwork}, where the absence of a preferred time slice presents a major challenge. In the static case, graph models provide a natural discretization of bulk geometry; extending this picture to dynamical spacetimes could provide a geometric framework for time-dependent tensor network constructions.

\subsection{Central question}
The central question we address is the following:
\begin{center}
    \emph{Can one construct a graph model for general covariant holographic states such that the discrete min-cut reproduces HRT entropies?}
\end{center}

A natural approach is to partition the bulk spacetime using entanglement wedges and their horizons. However, unlike the static case, HRT surfaces associated with different boundary regions need not lie on a common Cauchy slice. This raises a fundamental obstruction: a graph cut constructed from partial segments of HRT surfaces may form an unphysical ``short-cut'', whose total weight is smaller than the area of any homologous HRT surface.

This issue is highlighted in a recent study \cite{grado2025minimax}, which attempts to resolve the issue by proposing a minimax formulation of HRT entropy\cite{HH2023minimax}, dual to the maximin formulation of ref. \cite{wall2014maximin}. However, this approach requires one to impose a strong "cooperating property" on candidate timelike sheets to construct a graph model.

\subsection{Strategy and contributions}
In this work, we construct a graph model for covariant holographic states under a natural geometric condition, namely the existence of exposed regions for interacting HRT surfaces. Within this regime, we prove a Conditional No-Short-Cut Theorem ensuring that graph cuts cannot beat HRT surfaces

The key idea is to exploit the causal structure of entanglement wedges. We analyze how partial segments of HRT surfaces can be projected along entanglement horizons, and show that any putative short-cut can be deformed into a non-optimal surface without increasing the area. This argument relies on entanglement wedge nesting \cite{wall2014maximin} and focusing properties of null congruences \footnote{We provide a review of the focusing argument in Appendix \ref{app:focusing}.}, and is closely related to techniques used in proofs of connected wedge theorems \cite{may2020holographic,zhao2026proof,zhao2025beyond,lima2025sufficientGCWT}.

We first analyze the two-horizon configuration and identify the existence of a family of admissible weight functions, which corresponds to huge freedom in choosing the projection locus. 

In generalizing the-horizon construction to arbitrarily many boundary regions, we identify the following two regimes:
\begin{itemize}
    \item In the \emph{achronal regime}, projection loci can be chosen to lie on mutually achronal portions of entanglement horizons, allowing all relevant HRT surface to be projected on a common Cauchy slice.
    \item In the \emph{timelike cluster regime}, causal relations between HRT surfaces force a nested structure, which we treat collectively by transporting sections along piecewise-null congruences. we show that these induced sections corresponding to admissible weight functions that was identified in the two-horizon construction.
\end{itemize}

Another key ingredient of our proof is Lemma \ref{lemma:HRT_causal_better}, which highlights that any homologous minimal surface that has less area than the corresponding HRT surface must be causally related to the HRT surface. This allows us to use an argument by contradiction.

However, a key geometric issue arises from the way different HRT surfaces interact in spacetime. 
For a pair of HRT surfaces, one can define an associated interaction region consisting of points causally connected to both surfaces. 
In general, such interaction regions for different pairs may overlap or even be nested.

We will show that a crucial simplifying condition is the existence of \emph{exposed regions}, namely portions of an interaction region that are not covered by others. 
When such exposed regions exist, one can choose projection loci that are mutually achronal, enabling a consistent projection construction.

Combining these results, we establish the Conditional No-Short-Cut Theorem under the exposed-region condition, and show how timelike clustering partially extends the construction beyond this regime.

\begin{theorem}[Conditional No-Short-Cut Theorem]\label{thm:no-short-cut}
Let \(A_1,\cdots,A_n\) be $n$ spacelike-separated, disjoint boundary regions in a spacetime $M$ that satisfies Assumption \ref{assumption:1}.
Let $(\mathcal V,E,w)$ be a graph model constructed from our projection methods.

Assume that for every pair of HRT surfaces, the corresponding exposed region is nonempty (possibly after reduction via timelike clusters).

Then for any cut $W\subset \mathcal V$ homologous to a boundary region $A_I, I\subseteq \nset$, we have
\begin{equation}
|\gamma| \leq |C(W)|,
\end{equation}
where $\gamma$ is the HRT surface of $A_I$ and $|C(W)|$ denotes the weight of the cut $W$.

In particular, the minimal cut coincides with the HRT entropy.
\end{theorem}
As a consequence, the projection-based graph model correctly reproduced holographic entropies for all configurations satisfying the exposed-region condition.
This establishes the equivalence of the covariant and static holographic entropy cones within this regime.

\subsection{Organization of the paper}
The paper is structured as follows. In Section \ref{sec:review_graph_model}, we present a review of graph models for the RT cone. We then present a universal characterization of any graph model for the HRT entropy cone in Section \ref{sec:characterization_graph}. In particular, we identify two equivalent characterizations for the existence of a graph model in Theorem \ref{thm:polygon_inequality}. 
In Section \ref{sec:projection_graph_model}, we present our construction of a graph model based on projection arguments. Some key geometric results are summarized in Section \ref{sec:geometric_result}. The two-horizon configuration is shown in Section \ref{sec:two-horizon}. The construction for general configurations is shown in Section \ref{sec:general_projection}. Some technical proofs are presented in Section \ref{sec:tech_proofs}. We conclude and discuss in Section \ref{sec:conclusion}.

\subsection{Notations and Assumptions}\label{subsec:notation_assumptions}
Here we summarize the notations, conventions, and assumptions used throughout this paper.

We adopt natural units with $\hbar=c=1$ and set the AdS length scale $l_{\text{AdS}}=1$, while keeping Newton's constant $G_N$ explicit. Our notation follows ref. \cite{waldGR}, using the mostly-plus metric signature.

\begin{itemize}
\item \textbf{Spacetime regions:} Bulk regions are denoted by script letters ($\mathcal{U}, \mathcal{V}, \mathcal{W}, \cdots$), while boundary regions use straight capitals ($U, V, W, \cdots$). The same symbol may denote either a causal diamond or its Cauchy surface, with the meaning clear from context.

\item \textbf{Cauchy slices:} Bulk Cauchy slices are denoted by $\Sigma$ with appropriate subscripts, boundary Cauchy slices by $\hat{\Sigma}$ with subscripts. By abuse of notation, $\Sigma$ may also refer to Cauchy slices of the conformally compactified spacetime.

\item \textbf{Causal structure:} The bulk causal future/past of region $\mathcal{V}$ is $J^\pm[\mathcal{V}]$; for boundary region $V$, we write $J^\pm[V]$ for bulk causal influence and $\hat{J}^\pm[V]$ for boundary causal influence.

\item \textbf{Domains of dependence:} The bulk domain of dependence of $\mathcal{V}$ is $\mathcal{D}[\mathcal{V}]$; the boundary domain of dependence of $V$ is $\hat{D}[V]$. The bulk future and past horizons of a causal domain $V$ is $\HH^\pm[V]$.

\item \textbf{Entanglement structures:} For boundary region $V$, we denote the entanglement wedge by $\ew(V)$, and HRT surface by $\text{HRT}(V)$.

\item \textbf{Complements:} The set-theoretic complement within a Cauchy slice uses superscript $c$.
\end{itemize}

\begin{assumption}\label{assumption:1}
We assume throughout that:
\begin{enumerate}
\item The bulk spacetime $M$ satisfies the null curvature condition;
\item HRRT surfaces can be found via a maximin \cite{wall2014maximin};
\item The spacetime is AdS-hyperbolic (the conformal compactification $\overline{M}=M \cup \partial M$ admits a Cauchy slice);
\end{enumerate}
\end{assumption}

\section{Graph models for holographic entropy}\label{sec:review_graph_model}

\subsection{Review: Graph models in the static (RT) case}
We begin by reviewing the graph model construction for static holography \cite{bao2015holographic}, which provides the conceptual foundation for our covariant generalization.

Consider a bulk spacetime that admits a time-reflection symmetry, so that all relevant RT surfaces lie on a common Cauchy slice $\Sigma$. Let $A_1,\cdots, A_n$ be spacelike-separated, disjoint \footnote{As always, we allow boundary regions to share spatial boundaries.} boundary regions. Let $\RT(A_I)$ denote the least-area surface homologous to the composite region $A_I, I\subseteq \nset$.

The collection of all RT surfaces partitions the Cauchy slice $\Sigma$ into finitely many connected bulk regions. One then construct a graph $(\mathcal{V},E,w)$ as follows:
\begin{itemize}
    \item Each connected bulk region defines a vertex $v\in \mathcal{V}$.
    \item Two vertices are connected by an edge $e\in E$ if the corresponding bulk regions share a portion of an RT surface.
    \item The weight of an edge $w(e)$ is given by the area of the corresponding portion of the RT surface.
    \item Boundary vertices are labeled by the boundary regions $A_i$, to which they are adjacent.
\end{itemize}

A cut of the graph is a partition $\mathcal{V}= W\cup W^c$.
This corresponds to a hypersurface in the bulk assembled from partial entanglement horizons. The set of edges crossing this partition is $C(W) = \{(w, w') \in E \mid w \in W, w' \in W^c\}$, and its total weight $|C(W)| = \sum_{e \in C(W)} w(e)$ is assigned to cut $W$ as the cut weight.

We say that the cut $W$ is \emph{homologous} to a boundary region $A_I = \bigcup_{i\in I} A_i$ if the set of boundary vertices contained in $W$ is precisely those labeled by indices in $I$. Then we define a discrete entropy for any graph model.
\begin{definition}[Discrete Entropy]
For a graph model $(\mathcal{V}, E, w)$ with boundary coloring $b:\partial \mathcal{V}\to \nset$, the discrete entropy of a boundary region $A_I$ is
\begin{equation}
S^*(I) = \min_{\substack{\mathcal{V} = W \cup W^c, \ \partial W = I}} |C(W)|,
\end{equation}
where the minimization is over all cuts $W$ such that the set of boundary vertices contained in $W$ is precisely those colored by indices in $I$, i.e., $\partial W := \partial \mathcal{V} \cap W = \{v \in \mathcal{V} \mid b(v) \in I\}$.
\end{definition}

A key result proved in ref. \cite{bao2015holographic} is that the discrete entropy of a graph model recovers the RT entropy
\[S^*(I)=S(I),\]
where $S(I)$ denotes the RT entropy of boundary region $A_I$.

\subsection{Covariant extension: challenges and setup}
We now turn to the fully covariant setting, where entropies are computed by HRT surfaces \cite{HRT2007covariant}. In contrast to the static case, HRT surfaces for different boundary regions need not lie on a common Cauchy slice, and therefore do not define a canonical spatial partition of the bulk. 

A natural replacement for RT surfaces in this setting is provided by entanglement wedge. For a boundary region $A$, its entanglement wedge $\ew(A)$ is bounded by the union of its HRT surface and the associated future and past entanglement horizons:
\[
\pew(A) = \HH^{+}[A] \,\cup\, \HHRT(A) \,\cup\, \HH^{-}[A],
\]
The collection of all such entanglement horizons partitions the bulk spacetime into finitely many regions.

This suggests constructing a graph model directly from this spacetime partition. As in the static case:
\begin{itemize}
    \item vertices correspond to bulk regions defined by the partition;
    \item edges correspond to shared portions of entanglement horizons;
    \item boundary vertices are labeled by boundary regions.
\end{itemize}

However, a crucial new issue arises in defining the edge weights. A naive prescription would assign to each edge the area of the corresponding portion of an HRT surface where it intersects another entanglement wedge. While this prescription correctly reproduces entropies for certain unions of regions, if fails in general \footnote{We thank Guglielmo Grimaldi, Matthew Headrick, and Veronika E. Hubenyb for pointing this out.}.

The underlying problem is that the naive prescription partitions HRT surfaces at sections that are generically spacelike separated. However, the graph model is meant to capture holographic entropies after effectively ``integrating out'' or ``forgetting'' the time direction. In such a reduction, one would expect the partitioning sections on different HRT surfaces to be causally related. 

This reasoning suggests that the correct partitioning of HRT surfaces should be defined using sections that are causally related, rather than merely spatially intersecting. In Section~\ref{sec:two-horizon}, we build toward this understanding by first using null generators of entanglement horizons to align the relevant HRT segments. This ensures that HRT surfaces can be consistently compared and assembled into admissible bulk surfaces.

\begin{remark}[Use of $A^c$ and entanglement horizons]\label{remark:Ac_A}
In this work, the notation $\pew(A)$ denotes the union of the HRT surface $\HHRT(A)$ and the associated future and past null sheets:
\[
\pew(A) = \HH^+[A] \cup \HHRT(A) \cup \HH^-[A].
\]

We will also use the notation $\pew(A^c)$ to denote the complementary pair of null sheets emanating from the same extremal surface $\HHRT(A)$. This is a notational convenience: in general, we do not assume that $\HHRT(A)=\HHRT(A^c)$, and our arguments do not rely on the purity of the boundary state.

Thus, the labels $A$ and $A^c$ should be understood as referring to two complementary choices of null directions from a given extremal surface, rather than to entanglement wedges of complementary boundary regions.
\end{remark}
\section{Local characterization of graph models}\label{sec:characterization_graph}
In this section, we will present a global-local equivalence theorem on universal characterizations of HRT entropy graph models.
\subsection{Global-Local equivalence}
As noted above, any graph model of the holographic entropy cone must satisfy the condition that
\begin{equation}\label{eq:discrete_HRT}
    S^*(I)=S(I),
\end{equation}
where $S^*$ denotes the discrete/graph entropy while S denotes the HRT entropy.

Since \eqref{eq:discrete_HRT} must hold for all boundary regions $I\subseteq \nset$, it appears natural to impose the following assumption:
\begin{definition}[Weight assumption on connected components of HRT surfaces]
    We assume that the weight of any connected component of an HRT surface equal to its area.
\end{definition}\label{assump:weight_complete_HRT}

Then assigning the weight function $w:\Vset\to \R_{\geq 0}$ reduces to partitioning each connected HRT component into pieces. Below, we will refer to \eqref{eq:discrete_HRT} under the ``weight=area'' assumption on connected HRT components as the \emph{global compatibility condition}. We now argue that the global compatibility condition can be reduced to a set of polygon inequalities for connected unions of cells.

\paragraph{Polygon inequalities for connected-regions}
Let $\mathcal{R}$ be a connected union of bulk cells. 
Its boundary $\partial \mathcal R$ consists of a collection of HRT pieces \footnote{We adopt the convention that the strong subadditivity inequality involves three boundary regions instead of two boundary regions. That is, intersections of boundary regions are counted as a distinct boundary regions. With this convention, a boundary cell only contains a complete HRT surface. We will discard this trivial case in the following discussion.}.
A \emph{side} of $\mathcal R$ is defined as a subset of $\partial \mathcal R$ lying on a single HRT surface; it may consist of multiple edges of the graph as long as they belong to the same HRT surface. Furthermore, a side of a polygon could be disconnected if multiple disconnected components of an HRT surface are involved in the same $\mathcal R$ under consideration.

Consider any HRT side $\Gamma$ of a $\mathcal R$.  Completing $\Gamma$ to the complete HRT surface $\HHRT(A_I)$ on which it lies, we get a cut $C_\star$. On the other hand, composing $\partial \mathcal R \setminus \Gamma$ and $\HHRT(A_I) \setminus \Gamma$, we get another cut $C$. 
Obviously, the two cuts, $C$ and $C_\star$, are homologous. Then \eqref{eq:discrete_HRT} immediately implies that 
\begin{equation}\label{eq:polygon_ineq}
    w(\Gamma)\leq w(\partial \mathcal R\setminus \Gamma)=\sum_{\Gamma'\subset \partial \mathcal R,\Gamma'\neq \Gamma} w(\Gamma'), \quad \forall\, \Gamma' \subset \partial \mathcal R
\end{equation}
which states that any HRT side of a connected bulk region $\mathcal R$ has weight less than or equal to the sum of weights of the remaining sides. In particular, one can take $\mathcal R$ to be a cell/vertex $v$, therefore polygon inequalities \eqref{eq:polygon_ineq} need to hold for each individual cell.

If $\mathcal R$ has only three distinct HRT sides, \eqref{eq:polygon_ineq} is just the familiar triangle inequality that holds for any metric space.
In the RT case such inequalities are automatic, since all RT pieces lie on a common Cauchy slice and each RT piece $\Gamma$ is minimal among surfaces with the same boundary $\partial \Gamma$ \footnote{This can be seen through completing an RT piece $\Gamma$ into the HRT surface $\gamma$ on which it lies. If there exists any other competitor $\Gamma'$ with the same boundary as $\Gamma$, composing $\Gamma'$ and $\gamma\setminus \Gamma$ would yield a surface homologous to $\gamma$ and of smaller area than $\gamma$. This contradicts the definition of RT surface. The same argument implies that an HRT piece is of least area among surfaces with the same boundary on its maxmin slice.}. In the HRT case, the HRT pieces may lie on different slices, so the inequalities become nontrivial compatibility conditions.


We now show that the set of all polygon inequalities are equivalent to the global compatibility condition required for any graph model of HRT entropies.

\begin{theorem}[Local--global compatibility]\label{thm:polygon_inequality}
Let $C_\star$ denote the cut corresponding to the true HRT surface for a fixed boundary region $A_I, I\subseteq \nset$. 
Then the following are equivalent under the Weight Assumption \ref{assump:weight_complete_HRT}:
\begin{enumerate}
\item The cut $C_\star$ is globally weight-minimizing among all cuts with the same boundary homology.
\item Polygon inequalities \ref{eq:polygon_ineq} hold for all connected unions of cells;
\end{enumerate}
\end{theorem}

\begin{proof}
Let $\mathcal{R}$ denote a connected union of cells.
That $(1)$ implies $(2)$ follows from previous discussions by completing one particular HRT side of $\mathcal R$ to the corresponding HRT surface.

We now prove the other direction.
Let $C$ be any cut with the same boundary homology as $C_\star$. Let $\mathcal U$ and $\mathcal U_\star$ denote the corresponding sets of bulk cells. 
Consider the symmetric difference
\[
\mathcal V = \mathcal U \,\Delta\, \mathcal U_\star = (\mathcal U \setminus \mathcal U_\star) \cup (\mathcal U_\star \setminus \mathcal U).
\]
This set decomposes into connected unions of bulk cells.

For each connected component $\mathcal R\subset \mathcal V$, the two cuts $C$ and $C_\star$ intersect $\partial \mathcal R$ along \emph{complementary} sets of HRT pieces. 
Applying the connected-region polygon inequalities \eqref{eq:polygon_ineq} to $\mathcal R$, with $\Gamma$ being the side from $C_\star$, shows that replacing the portion of $C$ along $\partial \mathcal R$ with that of $C_\star$ does not increase the weight.
Performing this replacement for each connected component  $\mathcal{R}$ of $\mathcal{V}$ transforms $C$ into $C_\star$ without increasing the weight, proving
\[
w(C_\star) \le w(C).
\]
Thus $C_\star$ is globally minimizing.
\end{proof}

\begin{remark}
The theorem shows that global compatibility of any graph model for HRT entropies reduces entirely to verifying the local polygon inequalities associated with each connected regions.
Below we will verify the local polygon inequalities for our construction of a graph model.
\end{remark}

\subsection{Necessary condition: monotonicity}
\begin{figure}
    \centering
    \includegraphics[width=0.8\linewidth]{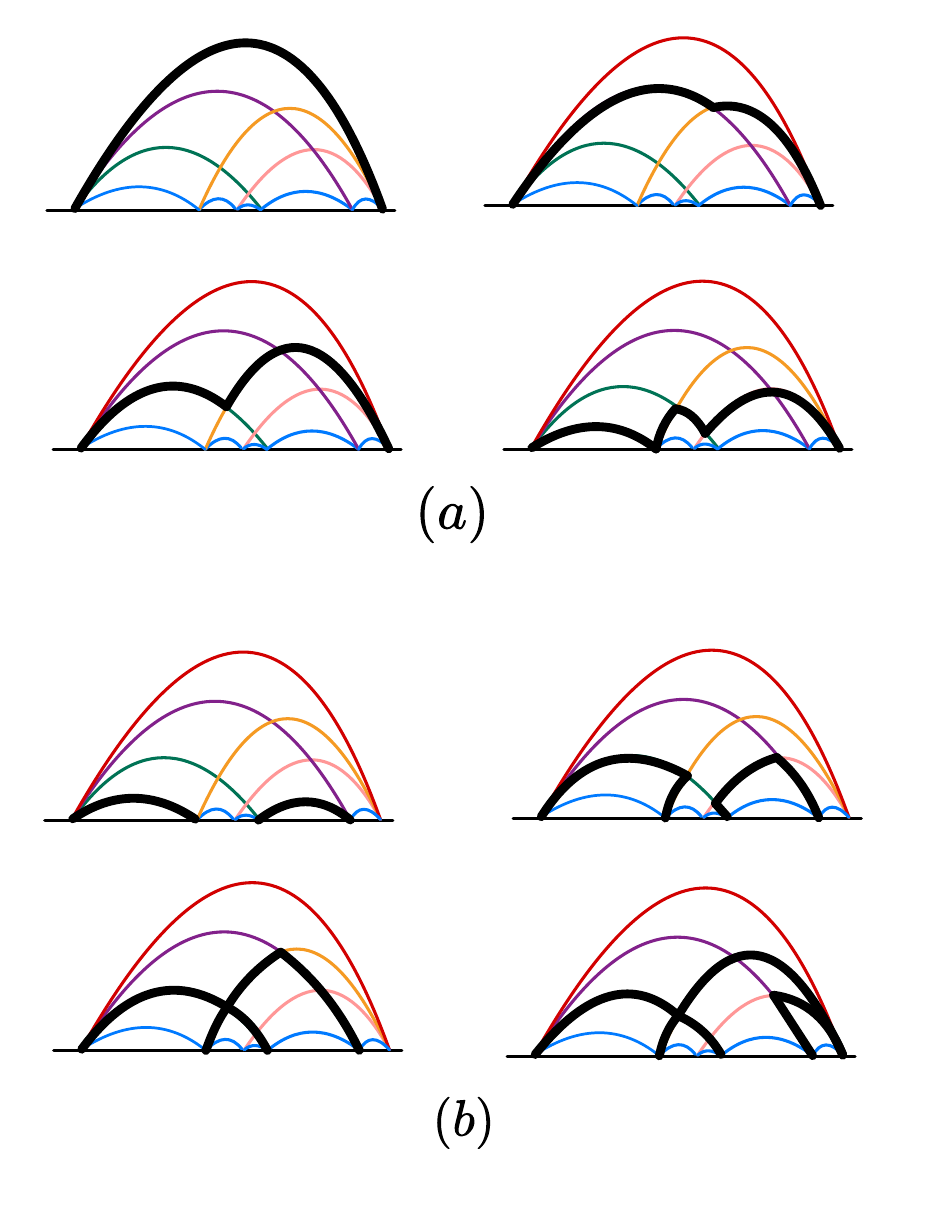}
    \caption{Illustration of monotonicity that is enforced by the global compatability condition of any graph model. Panel (a) illustrate that for a cut containing less vertices than the true HRT surface, the cut weight decreases as one includes more bulk vertices. Panel (b) illustrate that for a cut containing more vertices than the true HRT surface, the cut weight increases as one includes more bulk vertices.}
    \label{fig:monotonicity}
\end{figure}

A certain kind of monotonicity follows from polygon inequalities \eqref{eq:polygon_ineq}. Figure \ref{fig:monotonicity} illustrates two such examples that are present in a graph model.

Start from a cut $C$ that contains strictly less cells/vertices than $C_\star$. By adding more vertices or replacing multiple HRT sides with a single HRT side, the cut weight decreases monotonically (Figure \ref{fig:monotonicity}(a) in reverse order). Once $C=C_\star$, adding further vertices or replacing a single HRT side by multiple HRT sides, the cut weight increases monotonically (Figure \ref{fig:monotonicity}(b)). 
Since polygon inequalities \ref{eq:polygon_ineq} hold trivially in the RT case, so do these monotonicity observations. 

However, a monotonicity statement needs to be modified as the cut crosses an HRT surface or extremal surfaces corresponding to different connectedness of an entanglement wedge. Therefore, polygon inequalities are much cleaner formulations than a monotonicity statements. 


\begin{remark}
Importantly, the discussion in this section does not depend on the specific geometric realization of the graph. 
Instead, it relies only on the abstract graph structure. 
Consequently, the results in this section apply universally to any graph model intended to reproduce HRT entropies.
\end{remark}

\section{Projection-Based Graph Model}\label{sec:projection_graph_model}
Recall that polygon inequalities (or the more familiar triangle inequality) hold for any metric space
\begin{equation}\label{eq:metric_space_inequality}
    d(x_1,x_n)\geq d(x_1,x_2)+\cdots+d(x_{n-1},x_n),
\end{equation}
where $d(\cdot,\cdot)$ denotes the minimal distance/metric between two points of a metric space. The difficulty with HRT surfaces is that they do not necessarily lie on a common Cauchy slice and hence there is no natural metric space or metric $d(\cdot,\cdot)$. One natural approach would be to use projection method or focusing argument. Here we present results of such an approach.

The guiding principle is as follows:
\begin{itemize}
    \item To build a graph model for a configuration, we need to cut each HRT surface into partial pieces, i.e. edges. Then edges are assembled into cells/vertices. The critical contraint is that weight function for each partial HRT surface or edge is additive and that the total weight of a connected component of an HRT surface is its total area.
    \item Given a graph model, a ``short-cut'' may exist. To exclude such short-cuts using projection arguments, we need to map such short-cuts to spacelike geometric surfaces and argue that such geometric surfaces have greater area than the true HRT surface.
\end{itemize}

\subsection{Geometric Lemmas}\label{sec:geometric_result}
We collect here some geometric results that are essential in our construction.

We start with a simple observation that HRT surfaces intersect entanglement horizons simply. The corresponding fact for RT surfaces is established in ref. \cite{wall2014maximin} by cut and glue argument. We note that the cut and glue argument is a global argument that employs global minimality in area of RT surfaces while our proof only uses the local characterization, i.e. vanishing mean curvature, of HRT surfaces, ignoring that they are actually least-area extremal surfaces \footnote{Precisely, we use techniques for minimal surface equations while minimal surface in math literature refers to mean curvature vanishing or critical point of the area functional.}. Simpler proofs may exist.

\begin{figure}
    \centering
    \includegraphics[width=\linewidth,trim={10 14cm 10 10},clip]{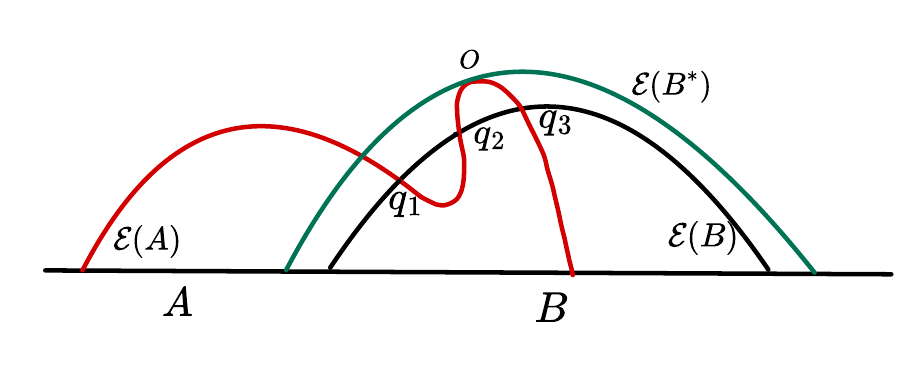}
    \caption{Illustration of no multiple entering of HRT surface into another entanglement wedge. The figure shows the Cauchy slice $\Sigma$ on which the HRT surface of A (red curve) is minimal. We also demand $\Sigma \cap \bb$ contains the spacelike boundary $\partial B$ and $\partial B^*$. The black and green curves denote $\partial \ew(B)\cap \Sigma$ and  $\partial \ew(B^*)\cap \Sigma$, respectively.} 
    \label{fig:HRT-cross}
\end{figure} 

\begin{lemma}[No Multi-Crossing of HRT Through Entanglement Horizon]\label{lemma:barrier}
Let $A,B \subset \partial M$ be boundary spatial regions on a common Cauchy slice of the boundary, and let $\gamma_A$ and $\gamma_B$ be their corresponding HRT surfaces in a classical asymptotically AdS spacetime satisfying the null energy condition (NEC) and the usual genericity assumptions. Denote the entanglement wedge of $B$ by $\ew(B)$ and its entanglement horizon by
\[
\partial\ew(B) = \gamma_B \cup \mathcal{H}^+[B] \cup \mathcal{H}^-[B].
\]
Then $\gamma_A$ can intersect $\partial\ew(B)$ at most once along each connected component. In particular, $\gamma_A$ cannot enter, exit, and re-enter $\ew(B)$.
\end{lemma}
Figure \ref{fig:HRT-cross} illustrates the HRT configuration that is excluded by Lemma \ref{lemma:barrier}. The proof is largely technical and does not rely on the main conceptual ingredients used in the main theorem below. We defer the proof to Section \ref{sec:proof-maximum-principal}.

\begin{lemma}[No partial coincidence with an entanglement wedge boundary]
\label{lemma:no_partial_coincidence}
Assume the same conditions as in Lemma \ref{lemma:barrier}. Then $\gamma_A$ cannot contain any nonempty open subset of $\partial \ew(B)$ unless $\gamma_A$ and $\gamma_B$ coincide on the corresponding connected component. In particular, for distinct HRT surfaces, $\gamma_A$ cannot partially lie along either $\gamma_B$ or the null components $\mathcal H^\pm[B]$ of $\partial\ew(B)$.
\end{lemma}

\begin{proof}
This follows by a similar argument using the strong maximum principle as in the proof of Lemma \ref{lemma:barrier}.
We give a short proof for the reader's convenience.

Suppose $\gamma_A$ contains a nonempty open subset of $\partial \ew(B)$.

If this open subset lies in $\gamma_B$, then $\gamma_A$ and $\gamma_B$ are two smooth extremal surfaces sharing an open subset. By the strong maximum principle after passing to a local spacelike slice, they coincide on the corresponding connected component.

If instead the open subset lies in one of the null components $\mathcal H^\pm[B]$, then $\gamma_A$ contains an open subset of a null hypersurface generated orthogonally from $\gamma_B$. Along $\mathcal H^\pm[B]$, the relevant null expansion satisfies $\theta \le 0$ by Raychaudhuri and the NEC. Since $\gamma_A$ is extremal, its null expansions vanish. Then the Raychaudhuri equation implies that the part of $\HH^\pm[B]$ between $\gamma_A$ and $\gamma_B$ also has vanishing expansion. Let $\gamma_B'$ be a section of $\HH^\pm[B]$ that partially coincides with $\gamma_A$. Then both $\gamma_B'$ and $\gamma_A$ have $\theta=0$ in the overllapping part. By the strong maximum principle after passing to a local spacelike slice, they coincide on the corresponding connected component. Thus, $\gamma_A$ lies entirely on $\HH^\pm[B]$.
\end{proof}

A direct consequence is that intersections between different pairs of entanglement horizons cannot share open subsets. As a result, we only need to consider intersecting entanglement horizons in pairs.

\begin{corollary}\label{cor:seam_pair}
Let $A_i,A_j$ be two distinct, spacelike-separated boundary regions, and define the intersection seam
\[
\mathcal S(A_i,A_j):=\partial \ew(A_i)\cap \partial \ew(A_j), \qquad i\neq j.
\]
Consider two distinct seams $\mathcal S(A_i,A_j)$ and $\mathcal S(A_k,A_l)$, i.e. at least one of $i,j$ differing from $k$ or $l$. They cannot share an open codimension-2 (relative to the bulk dimension) subset. Their intersection, if nonempty, must be of strictly higher codimension.
\end{corollary}

\begin{proof}
    This follows directly from the fact that $\cS(A_i,A_j)$ is of co-dimension $2$ in the bulk spacetime and hence can only have two null normals. If two distinct intersection seams $\cS(A_i,A_j)$ and $\cS(A_k,A_l)$ have open overlaps, then at least two distinct entanglement horizons associated with different boudnary regions would have to share an open family of null generators. Since the boundary regions are different, this would yield a contradiction to Lemma \ref{lemma:no_partial_coincidence}.
\end{proof}

\begin{remark}[Constraints on intersections of causal boundaries]
\label{remark:causal_boundary_intersection}
Different intersection seams can intersect, but their intersections are constrained by the causal properties of entanglement wedge boundaries. 

In particular, causal boundaries are achronal hypersurfaces, and therefore two causal boundaries cannot bound a compact bulk region. 
Such a configuration would imply that a null generator of one causal boundary connects two distinct points of the other causal boundary, which can be deformed into a timelike curve and hence violate achronality. 

Considering the future and past horizon of an entanglement horizon separately and then gluing at HRT surfaces, we get that the intersection seam between two entanglement horizons of $A$ and $B$ is a simple, continuous submanifold that ends on $\bb$. More specifically, the intersection seam ends on $\partial \hat{D}[A] \cap \partial \hat{D}[B]$, where $\hat{D}[\cdots]$ denotes a boundary domain of dependence.
\end{remark}

\subsubsection{Projection along global time function}\label{sec:proj_cauchy}
Here we highlight a key property about HRT surfaces, stated as Lemma \ref{lemma:HRT_causal_better}, which underlies the construction of our graph model.

Recall that the AdS-hyperbolicity of $M$, or equivalently the global hyperbolicity of $\overline{M}$, implies that $\overline{M}$ has the topology $\Sigma \times \mathbb{R}$, where $\Sigma$ is a Cauchy surface of $\overline{M}$. As established in Theorems 8.3.14 and 8.2.2 of \cite{waldGR}, global hyperbolicity ensures the existence of a global time function $t$ (though highly non-unique). Each level set of $t$ is a Cauchy surface $\Sigma_t$, and the gradient $\nabla t$ defines a global timelike vector field. By projecting along the integral curves of $\nabla t$, we can map all spacetime points to a fixed Cauchy slice $\Sigma_{t_0}$.

Moreover, using this global time function $t$ we can express one Cauchy slice as a graph over another Cauchy slice. Let $\Sigma_1$ and $\Sigma_2$ be two Cauchy slices such that
\[\Sigma_1\cap \bb=\Sigma_2\cap \bb.\]
We can regard $\Sigma_2$ as a graph of $t=f$ over $\Sigma_1$ either by choosing the time function $t$ that admits $\Sigma_1$ as a level set or by take $f=t_2-t_1$ where $t_i$ is the time coordinate of $\Sigma_i$ using any global time function $t$.

Let $\sigma_1$ be a submanifold of $\Sigma_1$. We have a corresponding submanifold $\sigma_2=\{t=f\}|_{\sigma_1}$ on $\Sigma_2$, which is just the restriction of the graph $\{t=f\}$ onto $\sigma_1$.

Let $ds_1^2=g_{ij} dx^i dx^j$ denote the metric on $\Sigma_1$, with $x_i$ being coordinates on $\Sigma_1$. Then $\Sigma_2$ has a natural metric
\begin{equation}
    ds_2^2=ds_1^2- df^2=(g_{ij}-f_if_j)\,dx^i dx^j,
\end{equation}
where $f_i=\frac{\partial f}{\partial x_i}$.
By restricting the metric to $\sigma_1$, we get that
\begin{equation}
    |\sigma_2|\leq |\sigma_1|,
\end{equation}
where $|\sigma_i|$ denotes the area of $\sigma_i$.

To see this more explicitely, let $\sigma_1$ be parameterized by coordiantes $y^\alpha$, with embedding $x^i=x^i(y)$. Then the induced metric on $\sigma_1$ is 
\[h^{(1)}_{\alpha\beta}=g_{ij}\,\frac{\partial x^i}{\partial y^\alpha}\frac{\partial x^j}{\partial y^\beta}.\]

Similarly, the corresponding submanifold $\sigma_2\subset \Sigma_2$ inherics the induced metric
\[h^{(2)}_{\alpha\beta}=(g_{ij}-f_if_j)\,\frac{\partial x^i}{\partial y^\alpha}\frac{\partial x^j}{\partial y^\beta}=h^{(1)}_{\alpha\beta}-f_\alpha f_\beta,\]
where $f_\alpha=\frac{\partial f}{\partial x^i}\frac{\partial x^i}{\partial y^\alpha}.$

Thus the induced metric on $\sigma_2$ differs from that on $\sigma_1$ by a negative semi-definite term. In particular, for any tangent vector $v^\alpha$ on $\sigma_1$,
\[h_{\alpha\beta}^{(2)}v^\alpha v^\beta=h_{\alpha\beta}^{(1)}v^\alpha v^\beta-(v^\alpha f_\alpha)^2\leq h_{\alpha\beta}^{(2)}v^\alpha v^\beta.\]
This implies that the volume element satisfies
\[\sqrt{\det h^{(2)}} \leq \sqrt{\det h^{(1)}}\]
and hence
\[|\sigma_2|\leq |\sigma_1|.\]

The reader may wonder if we instead express $\sigma_2$ as a graph over $\sigma_1$, we would get the opposite conclusion that \[|\sigma_1|\leq |\sigma_2|.\]
This contradiction is avoided by noting the following fact. If we use the same $t$ function then a graph of $t=-f$ is not $\Sigma_1$ but the reflection of $\Sigma_2$ with respect to $\Sigma_1$. This is because coordiantes $x^i$ is defined intrinsically on $\Sigma_1$. Instead if we take a different set of time function $t'$ and a different coordinates $(x')^i$ on $\Sigma_2$, the $\Sigma_1$ wouldn not be simply $\{t'=-f\}$.

The following simple example in $1+1$ Minkowski spacetime may illustrate this better.
Let
$$\Sigma_1=\{t=0\},\quad \Sigma_2=\{t = f(x) = \alpha x\},$$
where$|\alpha|<1$.
Then their corresponding metrics are:
\[ds_1^2=dx^2,\quad ds_2^2 = (1 - \alpha^2)\, dx^2.\]

Take
$\sigma_1 = \{x \in [-1,1]\}.$ Then $\sigma_2=\{(x,\alpha x)| x \in[-1,1]\}.$
Their areas are 
\begin{align*}
    |\sigma_1| &= \int_{-1}^1 dx = 2,\\
    |\sigma_2| &= \int_{-1}^1 \sqrt{1-\alpha^2}\, dx
= 2\sqrt{1-\alpha^2}.
\end{align*}
Thus,
$|\sigma_2| < |\sigma_1|$.

Now, we try to recover $\Sigma_1$ from $\Sigma_2$. If we take the old coordinate on $\Sigma_2$, i.e. $u=x$. Then 
$$\Sigma_2=\{(t,x)=(\alpha u, u)\}.$$
But $$\Sigma_1\neq \{(t,x)=(-\alpha u,u)\}.$$
Of course, one can rotate the coordinate such that $\Sigma_2=\{t'=0\}$. Then $\Sigma_1$ would not simply be $\{t'=-f\}$. So there is no contradiction by switching the role of $\Sigma_1$ and $\Sigma_2$.

\medskip

An application of the above argument to HRT surfaces yields the following crucial fact that any homologous minimal surface that has less area than its corresponding HRT surface must be causally related to the HRT surface.
\begin{lemma}\label{lemma:HRT_causal_better}
    Any smaller-area minimal surface that arises in the maximin formulation, e.g. a least-area surface within the specific homology class on a Cauchy slice that has strictly smaller area than the corresponding homologous HRT surface, cannot be spacelike separated from the HRT surface.
\end{lemma}
\begin{proof}
    We argue by contradiction. Let $\tilde \gamma$ be a minimal surface  arising in the maximin procedure that has less area than the corresponding homologous HRT surface $\gamma$ and is spacelike separated from the HRT surface $\gamma$. Since $\tilde\gamma$ and $\gamma$ are spacelike separated, there exists a Cauchy slice that contains both of them \footnote{See e.g. Lemma 14 of ref. \cite{HHLR2014} for a proof.}, which is denoted by $\tilde\Sigma$.
    Let $\Sigma$ denote a maximin Cauchy slice of the HRT surface $\gamma$. 

    Let $t$ be any global time function and regard $\Sigma$ as a graph over $\tilde\Sigma$, with graph function $t=f$.
    We can construct a surface $\tilde{\tilde{\gamma}}$ in the maximin slice $\Sigma$  by projecting $\tilde\gamma$ along the global time function $t$ as above. That is, $\tilde{\tilde{\gamma}}$ is a graph over $\tilde{\gamma}$, with the graph function being the restriction of $t=f$ to $\tilde{\gamma}$.

    Since the projection is along integral curves of a global time function, it defines a homotopy between $\tilde\gamma$ and $\tilde{\tilde{\gamma}}$ that does not intersect the boundary. Therefore, $\tilde{\tilde{\gamma}}$ remains in the same homology class as $\tilde\gamma$, and hence homologous to $\gamma$.
    
    By the above argument, projection along a timelike flow reduces area, so we have
    \[|\tilde{\tilde{\gamma}}|\leq |\tilde\gamma| < |\gamma|,\]
    which contradicts the fact that $\gamma$ is least-area in its homology class on $\Sigma$.
\end{proof}

\subsection{Two-horizon case}\label{sec:two-horizon}
\begin{figure}
    \centering
    \includegraphics[width=0.8\linewidth,trim={5 1cm 5 5},clip]{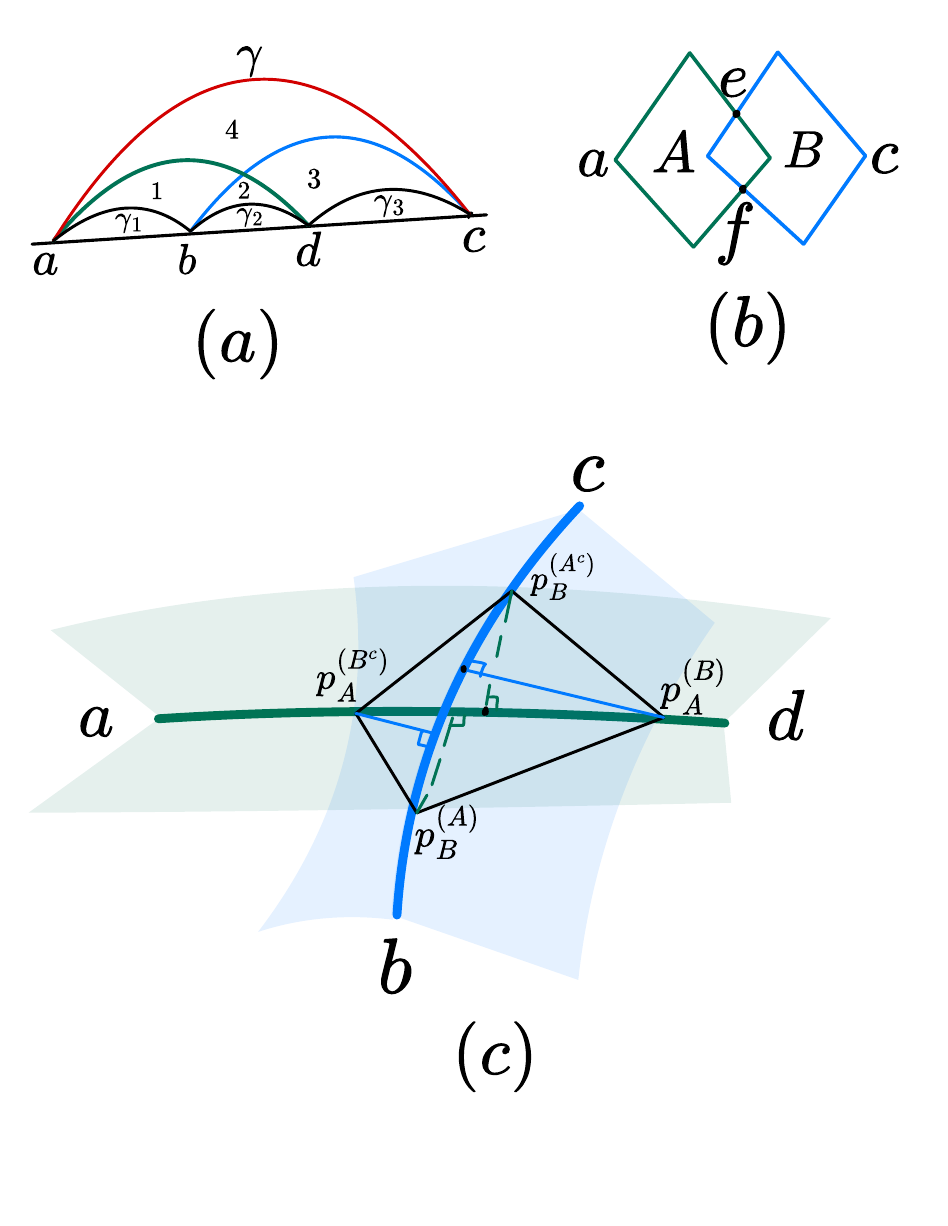}
    \caption{Geometry of two intersecting entanglement horizons. Panel (a) shows the setup ignoring the time direction with HRT$(A)$ and HRT$(B)$ shown in green and blue, respectively, $\gamma=\HHRT(A\cup B)$ shown in red, $\gamma_1=\HHRT(A\setminus B), \gamma_2=\HHRT(A \cap B), \gamma_3=\HHRT(B \setminus A)$ are shown in black. For later convenience we also label the four bulk cells. Panel (b) illustrates the two boundary domains of dependence $\hat{D}[A]$ (green) and $\hat{D}[B]$ and their intersections. Panel (c) illustrates the intersection geometry between null sheets emanating from $\HHRT(A)$ and from $\HHRT(B)$. Without loss of generality, we place $\HHRT(A)$ to the past of $\HHRT(B)$. We label where $\HHRT(A)$ crosses $\pew(B)$ and $\pew(B^c)$ as $p_A^{(B)}$ or $p_A^{(B^c)}$, respectively. Similarly for $p_B^{(A)}$ and $p_B^{(A^c)}$. To avoid clutter, we do not explicitly label the corresponding projecting sections $q_A^{(A\to B)}$, $q_A^{(A^c\to B)}$, $q_B^{(B\to A)}$, and $q_B^{(B^c\to A)}$, but indicate their locations using special markers (shown as $\perp$ symbols), which denote null generators from these sections reach projection loci.}
    \label{fig:two-horizon_geometry}
\end{figure}

Corollary \ref{cor:seam_pair} suggests that we could potentially reduce the problem to pairwise considerations. Therefore, we first study the configuration of two connected HRT surfaces associated with two intersecting boundary regions. 

Drawing an analogy to the RT case, it is natural to expect that the two-horizon configuration should be represented by a graph model with four bulk cells and that each HRT surface is partitioned into two edges. To understand the freedom in choosing where to partition each HRT surface into two edges, we study the geometry of intersecting null sheets associated with the two HRT surfaces, as illustrated in Figure \ref{fig:two-horizon_geometry}. 

\paragraph{Intersection Geometry of entanglement horizons}
We note that there are four null sheets emanating from each HRT surface $\HHRT(A)$, $\HH^\pm[A]$ and $\HH^\pm[A^c]$, where $A^c$ is used in the sense of Remark \ref{remark:Ac_A}.

Consider two boundary regions $A$ and $B$, each of which is of form $A_I, I \subseteq\nset$. Without loss of generality, we assume that $\HHRT(A)$ lies to the past of $\HHRT(B)$ near where their entanglement horizon intersects.

We denote where $\HHRT(A)$ crosses $\ew(B^c)$ by $p_A^{(B)}$ and where $\HHRT(A)$ crosses $\ew(B^c)$ by $p_A^{(B^c)}$.
The part of $\HHRT(A)$ bounded between the two sections $p_A^{(B)}$  and $p_A^{(B^c)}$ is causally related to $\HHRT(B)$ and we will refer to this part of $\HHRT(A)$ as $\gamma_A^{(B)}$. 

Similarly, we have two sections $p_B^{(A)}$ and $p_B^{(A^c)}$ on $\HHRT(B)$, defined to be where $\HHRT(B)$ crosses $\ew(A)$ and $\ew(A^c)$, respectively. Similarly the in-between part $\gamma_B^{(A)} \subset \HHRT(B)$ is causally connected to $\HHRT(A)$.

In short, any point in $\gamma_A^{(B)}$ can be connected by a causal curve to a point in $\gamma_B^{(A)}$. This suggests certain freedom in partitioning $\HHRT(A)$ and $\HHRT(B)$ into edges. 

Since the sections $p_{B}^{(A)},p_B^{(A^c)} \subset \HHRT(B)$ lie on future horizons of  $\HHRT(A)$, by well-known properties of causal boundaries, there exist sections $q_A^{(A\to B)}, q_{A}^{(A^c\to B)} \subset \HHRT(A)$ whose null generators along $\HH^+[A]$ and $ \HH^+[A^c]$ reach $p_{B}^{(A)}$  and $p_B^{(A^c)}$, respectively. Similarly, there exist $q_B^{(B\to A)}, q_B^{(B^c\to A)}\subset \HHRT(B)$ whose null generator reach $\HHRT(A)$ at $p_A^{(B)}$ and $p_A^{(B^c)}$, respectively.

Either by geometric intuition or by the causality argument deferred to Section \ref{sec:causality_position}, we have that $q_A^{(A\to B)}$ and $q_A^{(A^c\to B)}$ both lie in $\gamma_A^{(B)}$, the part of $\HHRT(A)$ causally connected to $\HHRT(B)$. Similarly, both $q_B^{(B\to A)}$ and $q_B^{(B^c\to A)}$ lie in $\gamma_B^{(A)}$. 

Note that there are four intersection seams associated with the two HRT surfaces: 
\begin{align*}
\cS(A,B)&=\pew(A)\cap \pew(B),\\
\cS(A,B^c)&=\pew(A)\cap \pew(B^c),\\
\cS(A^c,B)&=\pew(A^c)\cap \pew(B),\\
\cS(A^c,B^c)&=\pew(A^c)\cap \pew(B^c).
\end{align*}
For each intersection seam, we will only consider the portion bounded by two relevant HRT surfaces, namely the portion lying to the future of $\HHRT(A)$ and to the past of $\HHRT(B)$. In the following, the term ``intersection seam'' will always refer to this restricted portion.
These four seams bound two null quadrilaterals: one on the past horizons $\pew(B)\cup \pew(B^c)$ of $\HHRT(B)$, and one on the future horizons $\pew(A)\cup \ew(A^c)$ of $\HHRT(A)$.

\medskip

We will identify a family of weight functions that satisfy the global compatibility condition. As a first step, we define four \emph{canonical} weight functions using the above null related sections.
\begin{definition}[Four canonical weight functions]
There are four natural choices of partitioning the two HRT surfaces \footnote{If $q_A^{(A\to B)}=q_A^{(A^c\to B)}$ and $q_B^{(B\to A)}=q_B^{(B^c\to A)}$, then all arguments below simplify greatly. But this is a rather special situation.}:
\begin{enumerate}
    \item partition $\HHRT(A)$ into two edges where it crosses $\ew(B)$, i.e. $p_A^{(B)}$ while partition $\HHRT(B)$ at $q_B^{(B\to A)}$ whose null generators reach $p_A^{(B)}$.
    \item partition $\HHRT(A)$ into two edges where it crosses $\ew(B^c)$, i.e. $p_A^{(B^c)}$ while partition $\HHRT(B)$ at $q_B^{(B^c\to A)}$ whose null generators reach $p_A^{(B^c)}$.
    \item partition $\HHRT(B)$ into two edges where it crosses $\ew(A)$, i.e. $p_B^{(A)}$ while partition $\HHRT(A)$ at $q_A^{(A\to B)}$ whose null generators reach $p_B^{(A)}$.
    \item partition $\HHRT(B)$ into two edges where it crosses $\ew(A^c)$, i.e. $p_B^{(A^c)}$ while partition $\HHRT(A)$ at $q_A^{(A^c\to B)}$ whose null generators reach $p_B^{(A^c)}$.
\end{enumerate}
After partitioning an HRT surface using one of the four choices, we will conveniently denote the part that mostly lies out the other entanglement wedge by $\HHRT^o$ and the other part by $\HHRT^i$.    

Each partition induces a graph structure whose edges correspond to the resulting HRT segments. We assign to each edge a weight equal to the area of the corresponding HRT segment.
\end{definition}\label{defn:4case_cut_HRT}

In case $(1)$ and $(2)$, we project $\HHRT(B)$ along $\HH^-[B]$ or $\HH^-[B^c]$ onto the maximin slice of $\HHRT(A)$. By focusing, during this projection process, the area of both outer and inner parts of $\HHRT(B)$ do not increase. We then reduce to the RT case. In case $(3)$ and $(4)$, we instead project $\HHRT(A)$ toward a maximin slice of $\HHRT(B)$. In particular, we have the following theorem.

\begin{theorem}[Two intersecting entanglement horizons]\label{thm:2horizon_projection_4inequality}
Let \(A,B\) be boundary regions with overlapping domains of dependence
\(\hat D[A]\cap \hat D[B]\neq\varnothing\).
Then the entanglement horizons \(\pew(A)\) and \(\pew(B)\) intersect.
Decompose the two HRT surfaces using one of the four choices listed in Definition \ref{defn:4case_cut_HRT}:
\[
\mathrm{HRT}(A)=\mathrm{HRT}(A)^{\mathrm{o}}\cup \mathrm{HRT}(A)^{\mathrm{i}},
\qquad
\mathrm{HRT}(B)=\mathrm{HRT}(B)^{\mathrm{o}}\cup \mathrm{HRT}(B)^{\mathrm{i}}.
\]
Then the following inequalities are true:
\begin{align}
\bigl|\mathrm{HRT}(A)^{\mathrm{o}}\bigr|
+
\bigl|\mathrm{HRT}(B)^{\mathrm{o}}\bigr|
\;\ge\;
\bigl|\mathrm{HRT}(A\cup B)\bigr|, \label{eq:two-horizon-ineq1} \\
\bigl|\mathrm{HRT}(A)^{\mathrm{i}}\bigr|
+
\bigl|\mathrm{HRT}(B)^{\mathrm{i}}\bigr|
\;\ge\;
\bigl|\mathrm{HRT}(A\cap B)\bigr|,\label{eq:two-horizon-ineq2} \\
\bigl|\mathrm{HRT}(A)^{\mathrm{o}}\bigr|
+
\bigl|\mathrm{HRT}(B)^{\mathrm{i}}\bigr|
\;\ge\;
\bigl|\mathrm{HRT}(A\setminus B)\bigr|,\label{eq:two-horizon-ineq3} \\
\bigl|\mathrm{HRT}(A)^{\mathrm{i}}\bigr|
+
\bigl|\mathrm{HRT}(B)^{\mathrm{o}}\bigr|
\;\ge\;
\bigl|\mathrm{HRT}(B \setminus A)\bigr|. \label{eq:two-horizon-ineq4}
\end{align}
\end{theorem}

\begin{proof}

Note that the four boundary regions $ A\setminus B, B\setminus A, A\cap B$ and $A\cup B$ are either nested with or spacelike separated from $A$ or $B$.
Then by entanglement wedge nesting \cite{wall2014maximin}, we can choose a maximin slice $\Sigma_A$ of $\HHRT(A)$ or a maximin slice $\Sigma_B$ of $\HHRT(B)$ such that $\HHRT(A\setminus B), \HHRT(A \cap B), \HHRT(B\setminus A)$ and $\HHRT(A\cup B)$ are least-area among their respective homology classes on  $\Sigma_A$ or $\Sigma_B$.

For weight function of case $(1)$ and $(2)$ of Definition \ref{defn:4case_cut_HRT}, we project $\HHRT(B)$ onto $\Sigma_A$. Then we reduce to the RT case. 
For weight function of case $(3)$ and $(4)$ of Definition \ref{defn:4case_cut_HRT}, we project $\HHRT(A)$ onto $\Sigma_B$. Then we reduce to the RT case. 
The claim then follows directly.

\end{proof}

For the two-horizon configuration, we have twelve triangle inequalities in total four bulk cells. One can check easily that these twelve triangle inequalities are the only nontrivial polygon inequalities in this configuration.
We claim that the remaining eight triangle inequalities follow from the above \eqref{eq:two-horizon-ineq1}-\eqref{eq:two-horizon-ineq4}. 
\begin{corollary}\label{cor:triangel_trick}
    Assume the same conditions as in Theorem \ref{thm:2horizon_projection_4inequality}. Then we have four bulk triangular cells and the global compatibility condition is equivalent to twelve triangle inequalities (three inequalities for each cell).
    
    All twelve triangle inequalities are satisfied for each of the four weight function in Definition \ref{defn:4case_cut_HRT}. 
\end{corollary}

\begin{proof}
    We will only prove for case $(1)$ of Definition \ref{defn:4case_cut_HRT} and the proof for the other three cases is similar.

    Consider first the triangle with sides $\HHRT(A\cup B), \HHRT(A)^o, \HHRT(B)^o$, labeled as $4$ in Figure \ref{fig:two-horizon_geometry}. We have two triangle inequalities remain to prove for this cell. We do this by completing each partial HRT surface to its full HRT surface. We have
    \begin{align*}
        |\HHRT(A)^o| + |\HHRT(A)^i| \leq |\HHRT(B\setminus A)| + |\HHRT(A \cup B)| \\
        \leq |\HHRT(B)^o| + |\HHRT(A)^i| + |\HHRT(A \cup B)|,
    \end{align*}
    where the first inequality uses the fact that $\HHRT(B\setminus A) \cup \HHRT(A \cup B)$ is homologous to $\HHRT(A)$ and the second inequality uses the following proved triangle inequality in Theorem \ref{thm:2horizon_projection_4inequality}
    $$|\HHRT( B\setminus A)|\leq |\HHRT(A)^i| + |\HHRT(B)^o|.$$
    Subtracting $|\HHRT(A)^i|$ from both side, we get
    \begin{equation}\label{eq:Ao_small_4}
        |\HHRT(A)^o| \leq |\HHRT(A\cup B)| +|\HHRT(B)^o|.
    \end{equation}
    
All remaining triangle inequalities can be established using this completion trick from the four triangle inequalities \eqref{eq:two-horizon-ineq1}-\eqref{eq:two-horizon-ineq4} in Theorem \ref{thm:2horizon_projection_4inequality}.
\end{proof}
\begin{remark}
        This completion trick in the above proof would work as long as the complementary piece $\sigma_1$, e.g. $\HHRT(A) \setminus \HHRT(A)^o=\HHRT(A)^i$ for \eqref{eq:Ao_small_4} is a side of an adjacent triangular cell $v$ for which the following inequality is established 
        $$\sigma_3\leq \sigma_1 + \sigma_2,$$
        where $\sigma_2$ is the side shared by the two adjacent cells ($\HHRT(B)^o$ for \eqref{eq:Ao_small_4}) and $\sigma_3$ is the remaining side of triangle $v$ ($\HHRT(B\setminus A)$ for \eqref{eq:Ao_small_4}). 
\end{remark}

\medskip

We now extend the four canonical weight functions in Definition \ref{defn:4case_cut_HRT} to a family of admissible weight functions.

\begin{theorem}\label{thm:section_seam_cutHRT}
Take any section $s$ of one of the four intersection seams (in $2+1$ dimension, $s$ is just a point). Trace the section $s$ along null generators of $\pew(A)\cup \pew(A^c)$ to a section $p$ of HRT$(A)$. Meanwhile, trace the section $s$ along null generators of $\pew(B) \cup \pew(B^c)$ to a section $q$ of HRT$(B)$. Then define the weight function using sections $p$ and $q$, i.e. partitioning HRT$(A)$ and HRT$(B)$ at sections $p$ and $q$, respectively, into outer and inner parts and assigning respective areas of each portion as the weight function $w_s$. 
Varying section $s$, one obtains a family of weight functions $\{w_s\}$. 

Then each of such weight function yields all twelve triangle inequalities and hence satisfies the global compatibility condition for the two-horizon configuration.
\end{theorem}
\begin{proof}
We show that the four inequalities \eqref{eq:two-horizon-ineq1}-\eqref{eq:two-horizon-ineq4} of Theorem \ref{thm:2horizon_projection_4inequality} hold for a weight function $w_s$ obtained from a section $s$ of an intersection seam. Then Corollary \ref{cor:triangel_trick} establishes the remaining triangle inequalities. 

Note that the four boundary regions $ A\setminus B, B\setminus A, A\cap B$ and $A\cup B$ are nested or spacelike separated from $A$ or $B$.
Then by entanglement wedge nesting \cite{wall2014maximin}, we can choose a maximin slice $\Sigma_A$ of $\HHRT(A)$ and a maximin slice $\Sigma_B$ of $\HHRT(B)$ such that the four HRT surfaces -- $\HHRT(A\setminus B), \HHRT(A \cap B), \HHRT(B\setminus A)$ and $\HHRT(A\cup B)$ -- are least-area surfaces among their respective homology classes in both $\Sigma_A$ and $\Sigma_B$. We will use $\Sigma_A$ in the following.

Entanglement wedge nesting also implies that the defining section $s$ is spacelike separated from the four HRT surfaces -- $\HHRT(A\setminus B), \HHRT(A \cap B), \HHRT(B\setminus A)$ and $\HHRT(A\cup B)$, noting that section $s$ lies on $\pew(A)\cup \pew(A^c)$ and $\pew(B)\cup \pew(B^c)$. Therefore, there exists a Cauchy slice $\Sigma_s$ that contains both section $s$ and the four HRT surfaces.

Project both $\HHRT(A)$ and $\HHRT(B)$ along their respective horizons toward $\Sigma_s$. The area cannot increase during this projection by the focusing argument. Denote the projected image of $\HHRT(A)$ and $\HHRT(B)$ by $\gamma_A$ and $\gamma_B$, respectively. 
By construction, both projected surfaces pass through the section $s$, and hence $\gamma_A$ and $\gamma_B$ intersect at $s$.

Composing one of two pieces of $\gamma_A$ with one of two pieces of $\gamma_B$ at section $s$, we get four surfaces, each homologous to one of the four HRT surfaces $\HHRT(A\setminus B), \HHRT(A \cap B), \HHRT(B\setminus A)$ and $\HHRT(A\cup B)$. We will refer to these four glued surfaces as competitors. We then argue by contradiction that each competitor has area no smaller than the corresponding HRT surface.

Suppose by contradiction that a competitor has strictly less area than the homologous HRT surface. Then we project the competitor along any global time function onto the maximin slice $\Sigma_A$ as in Section \ref{sec:proj_cauchy}. This projection along time function further reduces area unless the projection is trivial.
\begin{itemize}
    \item If the projection is nontrivial, the projected surface has less area than the competitor and hence has strictly less area than the homologous HRT surface. 
    \item If the projection is trivial, then the competitor already lies on the maximin slice $\Sigma_A$, and hence provides a competing surface in the same homology class with strictly smaller area than the HRT surface. 
\end{itemize}
In either case, we arrives at a contradiction with the fact that an HRT surface is least-area among its homologous class on its maximin slice. 

The cut weight, defined as the sum of the areas of the relevant pieces of $\HHRT(A)$ and $\HHRT(B)$, is therefore greater than or equal to the areas of the four HRT surfaces $\HHRT(A\setminus B)$, $\HHRT(A \cap B)$, $\HHRT(B\setminus A)$, and $\HHRT(A\cup B)$.
\end{proof}

\begin{remark}
    Then Definition \ref{defn:4case_cut_HRT} corresponds to the special case of $s=p_A^{(B)}$, $s=p_A^{(B^c)}$, $s=p_B^{(A)}$ and $s=p_B^{(A^c)}$, respectively.
\end{remark}

\begin{definition}[Projection locus]
    Below we will refer to section $s$ in Theorem \ref{thm:section_seam_cutHRT} as projection locus.
\end{definition}

Theorem \ref{thm:section_seam_cutHRT} can be reformulated in the following way to highlight the freedom to define a weight function for the two-horizon configuration. This freedom will aid us in generalizing the construction to multiple horizon configurations.

We restrict to the subset of a HRT surface consisting of points whose null generators remain on the corresponding entanglement horizons until they reach the relevant intersection seams. The complement of this subset arises from caustics or generator endpoints and is expected to have measure zero. 

\begin{corollary}
Suppose a section $\alpha$ of $\HHRT(A)$ lies in the open set of points whose null generators along $\pew(A)\cup \pew(A^c)$ remain on the corresponding entanglement horizons until they reach the bounded intersection seams. Denote the resulting seam sections by
$$s_A^\alpha\subset \cS(A,B), \qquad s_{A^c}^\alpha \subset \cS(A^c,B).$$

Trace sections $s_A^\alpha$ and $s_{A^c}^\alpha$ along null generators of $\pew(B)\cup \pew(B^c)$ to sections of $\HHRT(B)$, denoted by $\beta_A$ and $\beta_{A^c}$, respectively.

Then, the two weight functions defined by partitioning $\HHRT(A)$ at section $\alpha$ while partitioning $\HHRT(B)$ at section $\beta_A$ and $\beta_{A^c}$, respectively, both satisfy the global compatibility condition. Denote the two weight functions by $w_{\alpha,{\beta_A}}$ and $w_{\alpha,{\beta_{A^c}}}$, respectively.

Moreover, any convex combination of the two weight functions $w_{\alpha,{\beta_A}}$ and $w_{\alpha,{\beta_{A^c}}}$ still satisfy the global compatibility condition. Geometrically, we partition $\HHRT(A)$ at section $\alpha$ while partition $\HHRT(B)$ at any intermediate sections between $\beta_A$ and $\beta_{A^c}$.

    The same conclusion holds by reversing the roles of the two HRT surfaces.
\end{corollary}\label{cor:section_HRT_cutHRT}
\begin{proof}
That both weight functions $w_{\alpha,\beta_A}$ and $w_{\alpha,\beta_{A^c}}$ satisfy the global compatibility condition follows from Theorem \ref{thm:section_seam_cutHRT} by taking $s=s_A^\alpha$ and $s=s_{A^c}^\alpha$, resepctively.

By Theorem \ref{thm:polygon_inequality}, the global compatibility conditions is equivalent to the set of all polygon inequalities. Since polygon inequalities are linear, any nonnegative linear combination of $w_{\alpha,\beta_A}$ and $w_{\alpha,\beta_{A^c}}$ also satisfy all polygon inequalities. However, to preserve Assumption \ref{assump:weight_complete_HRT}, namely that the total weight of each complete HRT component equals its area, we restrict to convex combinations, e.g.
$$w_\lambda=\lambda\, w_{\alpha,\beta_A} + (1-\lambda) \, w_{\alpha,\beta_{A^c}}, \quad \lambda \in [0,1].$$
Hence every $w_\lambda$ satisfies the global compatibility condition.

Any intermediate section between $\beta_A$ and $\beta_{A^c}$ determines a weight function whose values on the two resulting pieces of $\HHRT(B)$ lie between those of $w_{\alpha,\beta_A}$ and $w_{\alpha,\beta_{A^c}}$. Hence, at the level of weights, it can be written as a convex combination of the two endpoint weight functions. Therefore such intermediate geometric partitions also satisfy the global compatibility condition.
\end{proof}

\begin{remark}\label{remark:cut_timelike_points}
    Note that section $\alpha$ of HRT$(A)$ is timelike to section $\beta_A$ or $\beta_{A^c}$ of $\HHRT(B)$. One would naturally expect that the existence of a graph model for HRT entropy cone implies forgetting the time direction and reducing the dynamic HRT case to the static RT case. Therefore, Corollary \ref{cor:section_HRT_cutHRT} brings us closer to the expectation that we should partition the two HRT surfaces at causally connected sections.
\end{remark}

\subsection{General projection framework}\label{sec:general_projection}

We now extend the two-horizon projection construction to configurations involving multiple HRT surfaces.

Recall that our goal is to construct, for each potential short-cut from a graph, a geometric competitor built from pieces of HRT surfaces, and to compare its area with that of the corresponding HRT surface. The first step is to thus assemble HRT pieces into a continuous achronal geometric surface.

By Lemma \ref{cor:seam_pair}, no two distinct pairs of intersecting horizons share an open subset of intersection seams. Therefore, for any pair of HRT surfaces that share an intersection seam, we can locally apply the two-horizon projection construction of Theorem \ref{thm:2horizon_projection_4inequality} or Theorem \ref{thm:section_seam_cutHRT}. This allows us to project HRT pieces toward each other in pairs and produce a continuous geometric competitor associated to a given graph cut.

From Section \ref{sec:two-horizon}, the resulting surface is achronal in a neighborhood of each intersection seam, since it is obtained from locally projecting a pair of HRT surface toward a Cauchy slice, e.g. $\Sigma_s$. However, in the multi-horizon setting, this surface need not be globally achronal, especially when several HRT pieces are clustered together.

To resolve this concern, we examine more carefully the causal structure around a pair of HRT surfaces. For a pair $\HHRT(A)$ and $\HHRT(B)$, with $A,B=A_I$ for some $I\subset\nset$, we defined $\gamma_A^{(B)}\subset \HHRT(A)$ to be the portion lying between $\ew(B)$ and $\ew(B^c)$, and similarly $\gamma_B^{(A)}\subset \HHRT(B)$ to be the portion lying between $\ew(A)$ and $\ew(A^c)$. Then every point in $\gamma_A^{(B)}$ can be connected by a causal curve to a point in $\gamma_B^{(A)}$, and vice versa.

The four bounded intersection seams
\[
\cS(A,B),\quad \cS(A^c,B),\quad \cS(A,B^c),\quad \cS(A^c,B^c)
\]
bound a pair of null quadrilaterals, one on $\pew(A)\cup \pew(A^c)$ and one on $\pew(B)\cup \pew(B^c)$, which we denote by $\mathcal Q_A^B$ and $\mathcal Q_B^A$, respectively.

\begin{definition}[Interaction region]
For a arbitrary pair of HRT surfaces, e.g. $\HHRT(A)$ and $\HHRT(B)$. 
For a pair of HRT surfaces, we define the interaction region $\mathcal{C}_{A,B}$ to be the spacetime region bounded by the pair of null quadrilaterals, i.e.,
    $$\mathcal{C}_{A,B}=M\setminus \big(\ew(A)\cup \ew(A^c)\big) \, \bigcap \, M\setminus \big(\ew(B)\cup \ew(B^c)\big)=J^+[\HHRT(A)]\cap J^-[\HHRT(B)].$$
Note that \(\mathcal{C}_{A,B}\) contains all points that are causally connected to both $\HHRT(A)$ and $\HHRT(B)$.
\end{definition}

\begin{figure}
    \centering
    \includegraphics[width=\linewidth]{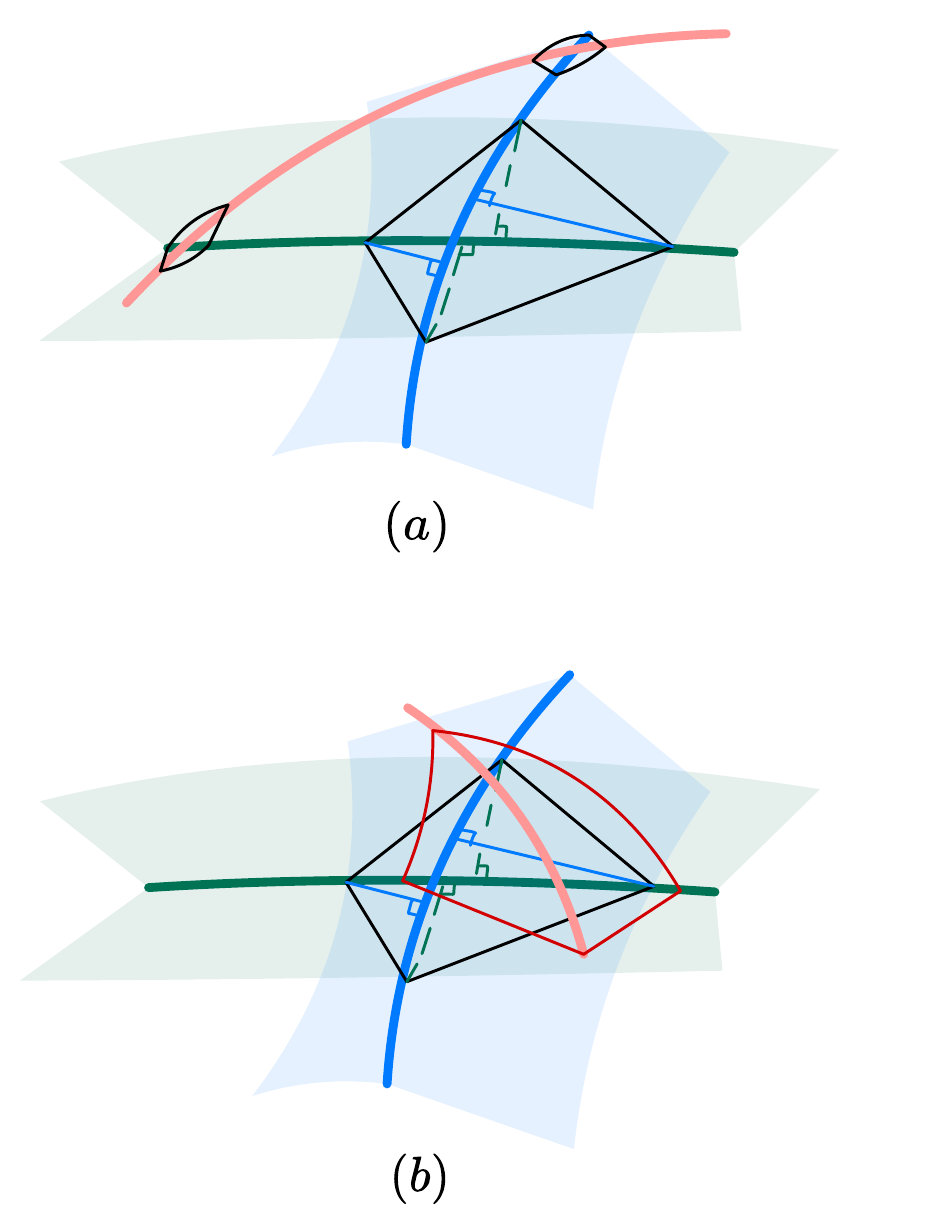}
    \caption{Illustration of disjoint vs overlapping null quadrilaterals. In Panel (a), three pairs of null quadrilaterals associated with three pairs of HRT surfaces are disjoint from each other. In Panel (b), the pair of null quadrilaterals associated with the pair of green and blue HRT surfaces largely overlap with the pair of null quadrilaterals associated with green and pink HRT surfaces. As a result, the triangular cell made of three HRT surfaces cannot not yield a achronal geometric surface.}
    \label{fig:favorable_unfavorable}
\end{figure}

 See Figure \ref{fig:favorable_unfavorable} for an illustration of disjoint and overlapping interaction regions associated with different pairs of HRT surfaces. We have the following key result that allows a reduction to pairwise consideration under a suitable assumption.

\begin{definition}[Exposed region]
    Let $\HHRT(B_i), i\in\{1,2,...,m\}$, with $m\geq 3$, be a collection of HRT surfaces. For a pair of HRT surface, define the \emph{exposed region} as
    \[E_{B_i,B_j}:=\mathcal{C}_{B_i,B_j} \setminus (\bigcup_{(k,l)\neq (i,j)} \mathcal{C}_{B_k,B_l}),\]
    where $(k,l)\neq(i,j)$ means one of $(k,l)$ is distinct from $i$ and $j$. This is the maximal region where a pairwise interaction is not masked by others.
\end{definition}
 
\begin{lemma}[Achronality from exposed loci] \label{lemma:achronality_disjointness}
    Suppose that $E_{B_i,B_j}\neq \emptyset$ for all pairs $(i,j), i,j\in \{1,\cdots, m\}$. 

    Choose projection loci $s_{ij}$ for each pair of HRT surfaces $(\HHRT(B_i),\HHRT(B_j))$ to lie inside the exposed region, i.e.
    \[s_{ij}\subset \partial \mathcal{Q}_{B_i}^{B_j} \cap E_{B_i,B_j}.\]
    Then all projection loci for this collection of HRT surfaces are mutually achronal.
\end{lemma}

\begin{proof}
We argue by contradiction. 

Without loss of generality, suppose $s_{12}$ for the pair $(\HHRT(B_1), \HHRT(B_2))$ is causally connected to another pair $(\HHRT(B_i),\HHRT(B_j))$, with $(i,j)\neq (1,2)$. 
Without loss of generality, we assume that locally $\HHRT(B_1)$ lies to the past of $\HHRT(B_2)$ and locally $\HHRT(B_i)$ lies to the past of $\HHRT(B_j)$. 

Let $l$ be a causal curve connecting a point $p\in s_{12}$ to a point $q\in s_{ij}$. We have two cases:
\begin{itemize}
    \item If $q\in J^+[p]$, extending $l$ on both ends by null generators yields a causal curve connecting a point in $\HHRT(B_1)$ to a point  in $\HHRT(B_j)$. That is, $l\subset \mathcal{C}_{B_1,B_j}$. Then 
    $$p\in s_{12}\cap \mathcal{C}_{B_1,B_j}\neq \emptyset,\quad q \in s_{ij}\cap \mathcal{C}_{B_1,B_j}\neq \emptyset.$$
    
    \item If $q\in J^-[p]$, extending $l$ on both ends by null generators yields a causal curve connecting a point in $\HHRT(B_2)$ to a point in $\HHRT(B_i)$. That is, $l\subset \mathcal{C}_{B_2,B_i}$. Then 
    $$p\in s_{12}\cap \mathcal{C}_{B_2,B_i}\neq \emptyset, \quad q\in s_{ij}\cap \mathcal{C}_{B_2,B_i}\neq \emptyset.$$ 
\end{itemize}
Both of the two cases contradicts the fact that $s_{12}\subset E_{B_1,B_2}$ and $s_{ij}\subset E_{B_i,B_j}$.



\end{proof}

If the condition of Lemma \ref{lemma:achronality_disjointness} is satisfied, i.e. exposed regions are nonempty, for all HRT surfaces involved in a connected union of cells $\mathcal R$, 
then the projection construction of Section \ref{sec:two-horizon} can be carried out to prove the set of polygon inequalities of $\mathcal{R}$. 
In particular, if all interaction regions for HRT surfaces in the n-boundary region problem are disjoint from each other, then the global compatibility condition directly follows from the pairwise discussion in Section \ref{sec:two-horizon}.

Below we first analyze the simple case that the condition of Lemma \ref{lemma:achronality_disjointness} can be satisfied and then discuss how to resolve the remaining case by defining timelike clusters.

\subsubsection{Achronal projection loci}\label{sec:achronal_loci}
In this subsection, we work under the assumption that all exposed regions are nonempty after possible reductions via timelike clustering. This will serve as the main condition under which our construction yields a valid graph model.

Let $\mathcal R$ be a connected union of cells \footnote{For our projection arguments, single cells are actually less trivial than nontrivial connected unions of cells.} with $m\geq 3$ HRT sides $\gamma_i \subset \HHRT(B_i)$, with $B_i=A_{I_i}, I_i\subseteq \nset$. Without loss of generality, consider the polygon inequality
\begin{equation}
    |\gamma_1| \leq \sum_{i=2}^m |\gamma_i|,
\end{equation}
or equivalently the inequality
\begin{equation}
    |\HHRT(B_1)| \leq |\HHRT(B_1)\setminus \gamma_1| + \sum_{i=2}^m |\gamma_i|.
\end{equation}

To establish polygon inequalities for $\mathcal{R}$, we only need to consider those pairs of HRT surfaces that are adjacent in $\mathcal{R}$. That is, we do not need to subtract all other interaction regions from a fixed $C_{B_i,B_j}$. However, the same HRT pair of surfaces may appear in other connected unions of cells. So we define the exposed region of each pair relative to the full collection of interaction regions in the $n$-boundary region problem. In this section, we assume all exposed regions are nonempty and we will propose a concept of timelike cluster to partially remove this assumption in Section \ref{sec:timelike_cluster}. 

We construct a geometric competitor $\sigma$ for $\HHRT(B_1)$ as follows.
We define a weight function by choosing projection loci $s$ inside exposed regions, as in Lemma \ref{lemma:achronality_disjointness}.
Then Lemma \ref{lemma:achronality_disjointness} implies that all projection loci invovled in $\mathcal{R}$ are achronal. Therefore, there exists a Cauchy slice $\Sigma$ that contains all such projection loci.
Then we project all HRT surfaces involved in $\mathcal R$ along their respective entanglement horizons toward $\Sigma$. By construction, the projected images intersect at projection loci. Our geometric competitor $\sigma$ is the image of $(\HHRT(B_1)\setminus \gamma_1) \,\bigcup\, \cup_{i=2}^m\gamma_i$.

The assembled surface $\sigma$ is continuous. Indeed, each segment of $\sigma$ is obtained by a continuous deformation of an HRT segment along null generators. By construction, these segments meet precisely at the chosen projection loci. The fact that each point on a horizon lies on precisely one null generator ensures that no conflicting identifications arise. Therefore the glued surface is a continuous codimension-2 surface.

Furthermore, $\sigma$ is homologous to the original HRT surface. This follows because each segment is obtained by deforming an HRT segment along causal curves that do not intersect the boundary, and the gluing preserves these deformations. Hence $\sigma$ is obtained from the original HRT surface by a homotopy that does not cross the boundary, and therefore lies in the same homology class.

In short, $\sigma$ a continuous, achronal codimension-2 surface that is homologous to $\HHRT(B_1)$. 
Then we argue by contradiction, as in Lemma \ref{lemma:HRT_causal_better}, to conclude that the area of $\sigma$ is no smaller than that of $\HHRT(B_1)$:

Suppose $\sigma$ has less area than $\HHRT(B_1)$. Then project $\sigma$ along any global time function $t$ toward a maximin slice $\Sigma_1$ of $\HHRT(B_1)$. If the projection is nontrivial, i.e. if $\sigma$ is not already in $\Sigma_1$, the projected surface $\tilde\sigma$ has less area than $\sigma$. Then $\tilde \sigma$ would have strictly less area than $\HHRT(B_1)$, which contradicts the definition of HRT surfaces. If the projection along time direction is trivial, i.e. $\sigma\subset \Sigma_1$, this again contradicts the definition of HRT surfaces.

\begin{remark}\label{remark:general_proj_achronal_surface}
The key ingredient in the above proof is that the projection loci can be chosen to be achronal such that a Cauchy slice contains all relevant ones.

Since the weight function is defined once and for all, the set of projection loci are also chosen once and for all.  Unless two projection loci belong to different components of the configuration, they can always appear in the same connected union of cells $\mathcal{R}$. For simplicity, we require simultaneous achronality of all projection loci.
\end{remark}

\subsubsection{Timelike cluster}\label{sec:timelike_cluster}
We now consider the remaining case that some exposed region $E_{B_i,B_j}$ is empty and hence it is not possible to ensure relevant projection loci are achronal to each other. This happens when some interaction region is contained in another interaction region or a union of other interaction regions. 

The key idea to resolve this is that nested or overlapping interaction regions induce an effective ordering among HRT surfaces, allowing them to be treated collectively as a single timelike cluster.

Intuitively, we would like to introduce a notion of timelike cluster as follows. 
\begin{definition}[Timelike cluster]
Suppose there exists a collection of HRT surfaces $\{\gamma_i\}$ for which there exists a congruence of timelike curves such that this congruence cut out a section on each of the HRT surface of this collection. In other words, each timelike curve of the congruence intersects every $\gamma_i$ exactly once. We refer to such a collection of HRT surfaces as a \emph{timelike cluster}.

For a timelike cluster, we partition all $\gamma_i$ simultaneously using the same parameter along the congruence. In particular, instead of cutting each $\gamma_i$ at its individual intersection with other entanglement wedge boundaries, we restrict to cuts induced by shared sections of the congruence.
\end{definition}

\begin{remark}
    However, Since our construction employs projection along null generators, we will not use this definition directly; instead we realize it via piecewise null congruences.
\end{remark}

For the sake of explanation, we define a chronological order among HRT surfaces. 
\begin{definition}[Chronological order]
    We write $\HHRT(A)\rightarrow \HHRT(B)$ if $\gamma_A^{(B)}$ and  $\gamma_B^{(A)}$ are nonemtpy \footnote{The two HRT segments must be nonemtpy simultaneously. If $\gamma_A^{(B)}\neq \emptyset$, then $p_A^{(B)}$ and $p_A^{(B^c)}$ are nonemtpy and correspondingly $q_B^{B\to A}$ and $q_B^{B^c\to A}$ are nonempty. Since  $q_B^{B\to A}$ and $q_B^{B^c\to A}$ are causally connected to $\HHRT(A)$, they lies in $\gamma_B^{(A)}$. Therefore, $\gamma_B^{(A)}\neq \emptyset$.} and $\gamma_A^{(B)}$ lies to the past of $\gamma_B^{(A)}$.
\end{definition}

Consider two different pairs of HRT surfaces $(\HHRT(B_1)\HHRT(B_2))$ and $\HHRT(B_i),\HHRT(B_j)$, with $(i,j)\neq (1,2)$ and the chronological order to be (always true by renaming)
\[\HHRT(B_1)\to \HHRT(B_2), \qquad \HHRT(B_i)\to \HHRT(B_j).\]

\paragraph{Nested interaction regions}
Suppose one of the interaction region is contained in the other. Without loss of generality, we take that
\[\mathcal{C}_{B_1,B_2} \subseteq \mathcal{C}_{B_i,B_j}.\]
It follows from the definition of interaction region, namely $\mathcal{C}_{B_i,B_j}=J^+[\HHRT(B_i)]\cap J^-[\HHRT(B_j)]$,
that
\begin{equation}\label{eq:gamma_contain_interaction}
    \gamma_{B_1}^{(B_2)}, \gamma_{B_2}^{(B_1)}\subset J^+[\HHRT(B_i)] \cap J^-[\HHRT(B_j)].
\end{equation}
Also note that 
\begin{align}\label{eq:interaction_region_def2}
    \mathcal{C}_{B_1,B_2} &=J^+[\HHRT(B_1)]\cap J^-[\HHRT(B_2)]=J^+[\gamma_{B_1}^{(B_2)}]\cap J^-[\gamma_{B_2}^{(B_1)}],\\
    \gamma_{B_1}^{(B_2)} &= \HHRT(B_1)\cap J^-[\HHRT(B_2)]=\HHRT(B_1)\cap J^-[\gamma_{B_2}^{(B_1)}]
\end{align}
These observations imply that 
\begin{align}
    \gamma_{B_1}^{(B_2)} &\subseteq \gamma_{B_1}^{(B_i)} = \HHRT(B_1)\cap J^+[\HHRT(B_i)], \label{eq:gamma_B1_B2_Bi} \\
    \gamma_{B_1}^{(B_2)} &=\HHRT(B_1)\cap J^-[\gamma_{B_2}^{(B_1)}]\subseteq \gamma_{B_1}^{(B_j)} = \HHRT(B_1)\cap J^-[\HHRT(B_j)],\label{eq:gamma_B1_B2_Bj} \\
    \gamma_{B_2}^{(B_1)} &\subseteq \gamma_{B_2}^{(B_j)}=\HHRT(B_2)\cap J^-[\HHRT(B_j)],\label{eq:gamma_B2_B1_Bi}  \\
    \gamma_{B_2}^{(B_1)} &=\HHRT(B_2)\cap J^+[\gamma_{B_1}^{(B_2)}]\subseteq \gamma_{B_2}^{(B_i)} =\HHRT(B_2)\cap J^+[\HHRT(B_i)]. \label{eq:gamma_B2_B1_Bj} 
\end{align}

We will treat this chain of causally connected HRT surfaces (if $i=1$ or instead $j=2$, then just collapse identical HRT surfaces)
\begin{equation}\label{eq:chronological_order}
    \HHRT(B_i)\to \HHRT(B_1)\to \HHRT(B_2)\to  \HHRT(B_j)
\end{equation}
as a timelike cluster.

What remains to be shown is that an admissible weight function can be defined such that it would satisfy the global compatibility condition.

It follows from  \eqref{eq:gamma_B1_B2_Bi}-\eqref{eq:gamma_B2_B1_Bj} that
we can construct a timelike congruence connecting these HRT surfaces using the following concatenated null generators.
\begin{itemize}
    \item First choose a seam section $s_{2j}$ associated with the HRT pair $\HHRT(B_2)$ and $\HHRT(B_j)$ such that its past-pointing null generators along $\HH^+[B_2]\cup \HH^+[B_2^c]$ land in $\gamma_{B_2}^{(B_1)}$. This is possible because $\gamma_{B_2}^{(B_1)}\subset \gamma_{B_2}^{(B_j)}$.
    
    Since $\gamma_{B_2}^{(B_1)}\subset \gamma_{B_2}^{(B_j)}$, future-pointing null generators from $s_{2j}$, along $\HH^-[B_j]\cup \HH^-[B_j^c]$, land in $\gamma_{B_j}^{(B_2)}$.
    
    Denote the cutting section determined by $s_{2j}$ on $\HHRT(B_2)$ and $\HHRT(B_j)$ by $\alpha_2$ and $\alpha_j$, respectively.
    
    \item Since $\alpha_2\subset \gamma_{B_2}^{(B_1)}$, it is connected through piecewise null generators, first along $\HH^-[B_2]\cup \HH^-[B_2^c]$ and then along $\HH^+[B_1]\cup \HH^+[B_1^c]$, to a section $\alpha_1\subset \gamma_{B_1}^{(B_2)}$ \footnote{Here we again assume null generators of $\alpha_2\subset \HHRT(B_2)$ remains on $
    \HH^-[B_2]\cup \HH^-[B_2^c]$ until reaching intersection seams with $\pew(B_1)\cup \pew(B_1^c)$. This can be guaranteed by varying the choice of the starting section $s_{2j}$ and the fact that failed sections have measure zero in $\HHRT(B_2)$.}.
    
    \item Since $\alpha_1\subset \gamma_{B_1}^{(B_2)}\subset \gamma_{B_1}^{(B_i)}$, it is connected through piecewise null generators, first along $\HH^-[B_1]\cup \HH^-[B_1^c]$ and then along $\HH^+[B_i]\cup \HH^+[B_i^c]$, to a section $\alpha_i\subset \gamma_{B_i}^{(B_1)}$ \footnote{A similar remark as in the previous item applies here.}.
\end{itemize}
We then define the respective weight as the relevant area or HRT pieces.
The last ingredient is to realize that the partitioning sections $\beta_1,\beta_2,\beta_i,\beta_j$ are among the intermediary sections of Corollary \ref{cor:section_HRT_cutHRT} and hence satisfy the global compatibility condition pairwise.

By construction, our weight function satisfies the global compatibility condition for adjacent pairs in 
\[    \HHRT(B_i)\to \HHRT(B_1)\to \HHRT(B_2)\to  \HHRT(B_j),\]
e.g. the pair $\HHRT(B_i$ and $\HHRT(B_1)$ etc. This is because sections of their bounded intersection seam are used explicitly in the above construction. From Section \ref{sec:achronal_loci} we know that the only condition to be imposed is these projection loci being achronal to those outside the timelike cluster.

For nonadjacent pairs like $(\HHRT(B_i)$ and $\HHRT(B_2))$, we note the following fact:
additional turning of our piece-wise null generators on intermediary entanglement horizons makes the congruence lie strictly inside relevant future/past null cones. For example, our congruence of piecewise null geodesics from $\alpha_2$, upon reaching $\HH^+[B_i]\cup \HH^+[B_i^c]$, lies strictly inside the intersection between $\HH^+[B_i]\cup \HH^+[B_i^c]$ and past null cones from $\beta_2$, i.e. between $\beta_{B_i}$ and $\beta_{B_i}^c$ in the notation of Corollary \ref{cor:section_HRT_cutHRT}, due to the additional turning along the entanglement horizons associated with $B_1$. This makes $\beta_i$ an intermediary section determined by $\beta_2$ as in Corollary \ref{cor:section_HRT_cutHRT}. We spell out this argument more rigorously in Section \ref{sec:intermediate_section}. 

However, for a weight function using an intermediary section to be admissible, the projection loci associated with the two endpoint sections ($\beta_A$ and $\beta_{A^c}$ in the notation of Corollary \ref{cor:section_HRT_cutHRT}) need to be achronal/compatible from other projection loci. 

We take care of this by removing projection loci associated with endpoint sections of any intermediary section in defining the new/effective exposed region associated with this timelike cluster.

Lastly, since concatenation of piecewise null causal curves are still piecewise null causal curves, the above discussion with two interactions regions directly extends to arbitrarily many interactions regions.

\medskip

\medskip

More generally, it could happen that $\mathcal{C}_{B_1,B_2}$ only overlaps with $\mathcal{C}_{B_i,B_j}$ but is contained in a union of interaction regions including $\mathcal{C}_{B_i,B_j}$.

We can define a timelike cluster provided those partitioning section $\alpha_1,\alpha_2,\alpha_i,\alpha_j$ can be defined.
In particular, we need
$$\gamma_{B_1}^{(B_2)}\cap \mathcal{C}_{B_i,B_j}\neq \emptyset, \quad \gamma_{B_2}^{(B_1)} \cap \mathcal{C}_{B_i,B_j}\neq \emptyset.$$


\medskip

While the timelike cluster construction resolves cases involving nested or certain overlapping interaction regions, it does not yet address configurations in which an interaction region is covered only by a union of others. We leave a complete treatment of such cases to future work.

In summary, we employ the timelike cluster construction to reduce the consideration to pairwise relations between clusters rather than between individual HRT surfaces. Compared to the naive construction -- where each HRT surface is partitioned independently at all intersection seams -- this clustering eliminates spurious cells involving timelike-related HRT surfaces. 

\subsection{Summary}



After reducing nested interaction regions via timelike clustering, our construction applies whenever all remaining exposed regions are nonempty. In this regime, the projection construction of Section \ref{sec:achronal_loci} establishes all polygon inequalities and hence the global compatibility condition.

However, it remains possible that some interaction region is not contained in any single other region but is covered by a union of them. In such cases, exposed regions may still be empty even after clustering. We leave the treatment of these configurations to future work.

Thus, our results establish the existence of a graph model under the exposed-region condition, and provide a partial extension beyond it via timelike clustering.

\section{Technical Proofs}\label{sec:tech_proofs}
We collect here some technical proofs that were described in the main text.
\subsection{Proof of Lemma~\ref{lemma:barrier}}\label{sec:proof-maximum-principal}
\begin{proof}[Proof of Lemma~\ref{lemma:barrier}]
If $A \cap B = \emptyset$, then by entanglement wedge nesting $\ew(A)$ is spacelike
separated from $\ew(B)$, hence $\gamma_A$ does not intersect $\partial\ew(B)$.

Assume $A \cap B \neq \emptyset$ and suppose for contradiction that along some connected
component of $\gamma_A$ there is an exit and a subsequent re--entry into $\ew(B)$. See Figure \ref{fig:HRT-cross} for an illustration.
As one traverses this component from one boundary anchor to the other, crossings with
$\partial\ew(B)$ must alternate between entering and exiting. Hence the existence of
more than one crossing implies the existence of three successive crossing components
\[
\mathfrak q_1,\ \mathfrak q_2,\ \mathfrak q_3 \subset \gamma_A\cap \partial\ew(B),
\]
corresponding to an entry into $\ew(B)$, an exit, and a re--entry.
Let $\mathcal{U}$ denote the open portion of $\gamma_A$ lying between
$\mathfrak q_2$ and $\mathfrak q_3$; by construction,
\[
\mathcal{U}\subset \mathrm{ext}\,\ew(B).
\]

\noindent\emph{First--contact tangency.}
Enlarge $B$ within the same boundary Cauchy slice to a one--parameter family $B(\lambda)$
with $B(0)=B$ and $B(\lambda_1)\supset B(\lambda_0)$ for $\lambda_1>\lambda_0$.
By entanglement wedge nesting, $\ew(B(\lambda_1))\supset \ew(B(\lambda_0))$.
For $\lambda=0$ we have $\mathcal{U}\subset \mathrm{ext}\,\ew(B(0))$.
Increase $\lambda$ until the first value $\lambda_*>0$ for which $\mathcal{U}$ completely lies inside $\partial\ew(B(\lambda_*))$, and set $B^*:=B(\lambda_*)$.

By this first--contact construction, there exists a point
\[
O\in \mathcal{U}\cap \partial\ew(B^*)
\]
such that $\partial\ew(B^*)$ is tangent to $\gamma_A$ at $O$, and in a sufficiently small
neighborhood of $O$ the surface $\gamma_A$ lies entirely on the same side of
$\partial\ew(B^*)$, namely the interior side of $\ew(B^*)$.
In particular, $\gamma_A$ intersects $\partial\ew(B^*)$ only at $O$ in that neighborhood.

\smallskip
\noindent\emph{Mean curvature sign on a spatial slice.}
The point $O$ lies on either $\mathcal{H}^+(B^*)$ or $\mathcal{H}^-(B^*)$.
Let $k^\mu$ be the inward--directed null generator of the corresponding horizon
component. By the Raychaudhuri equation and the null energy condition, the inward null
expansion satisfies $\theta\le 0$ along the generators, and by the usual genericity
assumption we may take the first--contact point so that
\[
\theta(O)<0.\footnote{\label{foot:weakMP}
An alternative argument avoids enlarging $B$ to $B^*$ and does not require the strict
inequality $\theta(O)<0$. One may instead apply a weak maximum principle to the portion
of $\gamma_A$ between $\mathfrak q_2$ and $\mathfrak q_3$ and the corresponding portion
of $\partial\ew(B)$ on a common Cauchy slice, provided both can be expressed locally as
graphs. We do not pursue this variant here.}
\]

Choose a bulk Cauchy slice $\Sigma$ containing $\gamma_A$.
Since $\gamma_A$ is extremal, its mean curvature vector in the full spacetime vanishes,
and hence $\gamma_A$ has vanishing mean curvature as a hypersurface in $\Sigma$.
Deform $\Sigma$ in a small neighborhood of $O$ (while containing $\gamma_A$) so that its second fundamental form becomes arbitrarily small near $O$.
Define
\[
S_1:=\gamma_A\subset\Sigma,
\qquad
S_2:=\partial\ew(B^*)\cap\Sigma.
\]
Let $\nu_{\mathrm{in}}$ denote the unit normal to $S_2$ in $\Sigma$ pointing inward into
$\ew(B^*)$. The standard relation between null expansion and spatial mean curvature gives
\[
\langle H_{S_2},\nu_{\mathrm{in}}\rangle
= \theta + \mathcal{O}(K^\Sigma),
\]
so for $\Sigma$ chosen as above,
\[
\langle H_{S_2},\nu_{\mathrm{in}}\rangle(O)<0.
\]

\smallskip
\noindent\emph{Local strong comparison.}
By the first--contact construction, $S_1$ lies locally on the interior side of $S_2$.
Working in local coordinates on $\Sigma$ near $O$, we may represent $S_1$ and $S_2$ as
graphs over a common domain $\Omega\subset\mathbb{R}^{d-1}$,
\[
S_1=\{(x,u_1(x))\},\qquad S_2=\{(x,u_2(x))\},
\]
with
\[
u_1\le u_2 \quad \text{in } \Omega,\qquad u_1(x_0)=u_2(x_0),
\]
where $x_0$ corresponds to $O$.

Let $\mathcal{M}$ denote the minimal surface operator for graphs with respect to the
downward normal, as reviewed in Appendix~\ref{sec:appendix_max_principle}.
Then $S_1$ minimal implies $\mathcal{M}(u_1)=0$, while the strict inequality
$\langle H_{S_2},\nu_{\mathrm{in}}\rangle(O)<0$ implies $\mathcal{M}(u_2)<0$ in a
neighborhood of $x_0$ (after possibly shrinking $\Omega$).
Thus locally,
\[
\mathcal{M}(u_1)\ge \mathcal{M}(u_2),\qquad u_1\le u_2,\qquad u_1(x_0)=u_2(x_0).
\]
By the strong comparison principle for $\mathcal{M}$
(Proposition~\ref{prop:strongcomp}), which is a purely local statement,
it follows that $u_1\equiv u_2$ near $x_0$.
This would imply that $\gamma_A$ locally coincides with $\partial\ew(B^*)$, 
which contradicts the first-contact set-up.

This contradiction shows that $\gamma_A$ cannot exit and re--enter $\ew(B)$ along the
same connected component. Equivalently, $\gamma_A$ intersects $\partial\ew(B)$ at most
once along each connected component.
\end{proof}

\subsection{Causality argument for position of intersection seams on entanglement horizons}\label{sec:causality_position}
We present here a rigorous argument for the position of an intersection seam relative to null generators of an entanglement horizon.

Recall that entanglement wedge $\ew(A)$ is defined as the domain of dependence of the homology region $\mathcal{R}_{hom}$ of $V$. Denote the set complement of the homology region of V on a suitable bulk Cauchy slice by $\mathcal{R}_{hom}'$. An important fact that we will use is
\begin{equation}\label{eq:horizon_homology}
    \HH^\pm[A]=\partial J^\pm[\mathcal{R}_{hom}'] \cap M,
\end{equation}
which states that the future and past entanglement horizons of $A$ is (the bulk part of) the causal future and past boundaries of  $\mathcal{R}_{hom}'$, respectively. It then follows from \eqref{eq:horizon_homology} and properties of causal boundaries that the future (past) null generator on the future (past) entanglement horizon $\HH^+[V]$ ($\HH^-[V]$) cannot enter the interior of the entanglement wedge but can possibly enter the exterior of the entanglement wedge \footnote{see Lemma A.1 in \cite{zhao2026proof} for a detailed handling of future (past) null generator on a future versus past causal boundary.}. The latter situation happens exactly when the null generator stops being prompt due to focusing etc. \footnote{A null generator of a causal boundary is said to be non-prompt when the null generator enters the interior of the causal future/past of the relevant set. We refer the reader to \cite{witten2020light} for a detailed discussion on promptness.}. Also by properties of causal boundaries e.g. achronality, null geodesic emanating from points in the exterior of an entanglement wedge cannot enter the corresponding entanglement wedge.

Direct applications of these observations reveal that the portion $\cS(A,B)^b$ of the intersection seam $\cS(A,B)$ bounded between the two HRT surfaces can only be connected to the outer portion of each HRT surface by null geodesics. We arrived at the same conclusion by simple geometric considerations as above. 

\subsection{Intermediate section property}\label{sec:intermediate_section}

\begin{lemma}[Intermediate-section property]
Let
\[
\HHRT(B_i)\to \HHRT(B_1)\to \HHRT(B_2)\to \HHRT(B_j)
\]
be a timelike-ordered chain arising from nested interaction regions, so that
\[
\mathcal C_{B_1,B_2}\subset \mathcal C_{B_i,B_j}.
\]
Let \(\alpha_2\subset \HHRT(B_2)\) be a section lying in
\(\gamma^{B_1}_{B_2}\). Transport \(\alpha_2\) backwards by the concatenated null
generators through \(\partial \ew(B_1)\cup\partial \ew(B_1^c)\) and then through
\(\partial \ew(B_i)\cup\partial \ew(B_i^c)\), obtaining a section
\(\alpha_i\subset \HHRT(B_i)\).

Then \(\alpha_i\) lies between the two endpoint sections on \(\HHRT(B_i)\) obtained
by applying the two direct transports from \(\alpha_2\) associated with the pair
\((B_i,B_2)\). Equivalently, \(\alpha_i\) is an intermediate section in the sense of
Corollary~\ref{cor:section_HRT_cutHRT}.
\end{lemma}

\begin{proof}
The nesting assumption
\[
\mathcal C_{B_1,B_2}\subset \mathcal C_{B_i,B_j}
\]
implies, in particular, the inclusions
\[
\gamma^{B_2}_{B_1}\subset \gamma^{B_i}_{B_1},
\qquad
\gamma^{B_1}_{B_2}\subset \gamma^{B_i}_{B_2}.
\]
Thus the portion of \(\HHRT(B_2)\) causally connected to \(\HHRT(B_1)\) lies inside
the portion of \(\HHRT(B_2)\) causally connected to \(\HHRT(B_i)\). In particular,
the section \(\alpha_2\subset \gamma^{B_1}_{B_2}\) also lies inside
\(\gamma^{B_i}_{B_2}\).

For the pair \((B_i,B_2)\), Corollary~\ref{cor:section_HRT_cutHRT} associates to
\(\alpha_2\) two endpoint sections on \(\HHRT(B_i)\), obtained by transporting
\(\alpha_2\) through the two relevant seams along
\[
\partial \ew(B_2)\cup \partial \ew(B_2^c)
\quad\text{and}\quad
\partial E(B_i)\cup \partial E(B_i^c).
\]
These two endpoint sections bound precisely the part of \(\HHRT(B_i)\) causally
connected to \(\alpha_2\) through the pairwise interaction region
\(\mathcal C_{B_i,B_2}\).

Now compare this direct transport with the concatenated transport through
\(\HHRT(B_1)\). The latter first follows null generators from \(\alpha_2\) to the
interaction region of the pair \((B_1,B_2)\), then turns along
\(\partial \ew(B_1)\cup\partial \ew(B_1^c)\), and finally follows null generators to
\(\HHRT(B_i)\). Since
\[
\mathcal C_{B_1,B_2}\subset \mathcal C_{B_i,B_2},
\]
this concatenated route stays inside the causal channel bounded by the two direct
transports associated with the pair \((B_i,B_2)\).

Equivalently, the extra turn along the intermediate horizons
\(\partial E(B_1)\cup\partial E(B_1^c)\) pushes the transported section into the
interior of the null strip bounded by the two direct horizon transports. By
achronality of the relevant causal boundaries, null generators cannot cross and then
return across one another. Therefore the endpoint section \(\alpha_i\) cannot lie
outside the interval bounded by the two direct endpoint sections on \(\HHRT(B_i)\).

Hence \(\alpha_i\) is an intermediate section between the two endpoint sections
determined directly by the pair \((B_i,B_2)\). By Corollary~\ref{cor:section_HRT_cutHRT},
partitioning \(\HHRT(B_i)\) at \(\alpha_i\) is admissible for the pair
\((B_i,B_2)\).
\end{proof}

\section{Conclusion}\label{sec:conclusion}

In this work, we established a graph model for holographic entropies in covariant spacetimes under a natural geometric condition, namely the existence of exposed regions for interacting HRT surfaces.

We began by identifying a universal criterion for the existence of such a model—the global compatibility condition—and showed that it is equivalent to a set of local polygon inequalities (Theorem~\ref{thm:polygon_inequality}). This reformulation reduces the problem of constructing a graph model to verifying local geometric constraints.

We then provided an explicit construction based on projection along entanglement horizons. In the two-horizon configuration, we identified a large family of admissible weight functions (Theorem \ref{thm:section_seam_cutHRT} and Corollary~\ref{cor:section_HRT_cutHRT}), reflecting a nontrivial freedom in choosing projection loci along null generators. We extended this construction to general multi-horizon configurations by distinguishing two regimes:
\begin{itemize}
    \item In the achronal regime, projection loci can be chosen to lie on a common Cauchy slice, and the construction reduces directly to the static RT case.
    \item In the timelike cluster regime, causal relations induce a nested structure among HRT surfaces. By transporting sections along piecewise-null congruences, we showed that these clusters admit compatible partitions inherited from the two-horizon construction.
\end{itemize}

Combining these ingredients, we proved the Conditional No-Short-Cut Theorem \ref{thm:no-short-cut}, establishing that any graph cut has weight greater than or equal to the area of the corresponding HRT surface. As a consequence, the minimal cut reproduces HRT entropies, and the covariant holographic entropy cone coincides with the static RT cone. In particular, polyhedrality and the finiteness of entropy inequalities extend to general time-dependent holographic states.

Conceptually, our construction shows that the role of a common Cauchy slice in the static case is replaced, in the covariant setting, by the causal structure of entanglement wedges. The graph model emerges not from spatial geometry, but from the organization of null congruences and their intersection structure.

Several directions remain open. It would be desirable to further clarify the geometric structure of timelike clusters and to formulate the construction in a more intrinsic way, independent of auxiliary choices of projection. Another important question is whether this framework can be extended beyond the classical regime, for example to include quantum extremal surfaces. Finally, the existence of a covariant graph model suggests a possible route toward constructing holographic tensor networks for time-dependent states, where the absence of a preferred time slice has been a longstanding obstacle.

We hope that the perspective developed here will provide a useful bridge between geometric, combinatorial, and quantum-information-theoretic approaches to holography.

\acknowledgments
I thank Edward Witten for introducing this problem to me and for helpful suggestions that improved the manuscript. 
I thank Matthew Headrick, Guglielmo Grimaldi, and Veronika E. Hubeny for pointing out a crucial gap in a previous version of this work and for helpful discussions. 
I thank Ning Bao for his interest in this work.

\renewcommand{\thesection}{\Alph{section}} 
\setcounter{Counter}{1}
\renewcommand{\thetheorem}{\Alph{section}.\arabic{theorem}} 

\appendix
\section{Focusing of null congruences and area monotonicity}\label{app:focusing}

In this appendix, we briefly review the focusing property of null congruences and its implication for area monotonicity along null hypersurfaces. This underlies the projection arguments used in the main text.

\subsection{Null congruences and expansion}

Let $\mathcal{N}$ be a null hypersurface generated by a congruence of null geodesics with tangent vector field $k^\mu$. The expansion $\theta$ of the congruence is defined as
\begin{equation}
\theta = \nabla_\mu k^\mu,
\end{equation}
which measures the infinitesimal rate of change of cross-sectional area along the generators.

More precisely, if $\sigma$ is a codimension-2 spacelike cross section of $\mathcal{N}$, and we flow it along the generators by an affine parameter $\lambda$, then its area element evolves as
\begin{equation}
\frac{d}{d\lambda} \log \sqrt{h} = \theta,
\end{equation}
where $h$ is the induced metric on $\sigma$.

\subsection{Raychaudhuri equation and focusing}

The evolution of $\theta$ along the congruence is governed by the Raychaudhuri equation:
\begin{equation}
\frac{d\theta}{d\lambda}
= -\frac{1}{d-2}\theta^2 - \sigma_{\mu\nu}\sigma^{\mu\nu}
- R_{\mu\nu} k^\mu k^\nu,
\end{equation}
where $\sigma_{\mu\nu}$ is the shear tensor and the twist term is omitted because the twist vanishes for hypersurface orthogonal congruences.

Assuming the null energy condition,
\begin{equation}
R_{\mu\nu} k^\mu k^\nu \ge 0,
\end{equation}
we obtain
\begin{equation}
\frac{d\theta}{d\lambda} \le 0.
\end{equation}
Thus, the expansion is non-increasing along null generators.

\subsection{Area monotonicity}

As a consequence, if $\theta \le 0$ at some cross section, then it remains non-positive along the future direction of the congruence. Integrating the expansion, we find that the area of cross sections is non-increasing along the flow:
\begin{equation}
|\sigma(\lambda_2)| \le |\sigma(\lambda_1)| \quad \text{for } \lambda_2 \ge \lambda_1.
\end{equation}

\subsection{Application to entanglement horizons}

In the context of holography, entanglement horizons $\HH^\pm[A]$ are generated by null congruences emanating from extremal surfaces. At the HRT surface, the null expansions vanish:
\begin{equation}
\theta = 0.
\end{equation}
Therefore, along the appropriate null directions, the expansion becomes non-positive, and the area of cross sections decreases away from the extremal surface.

This implies that projecting a surface along null generators of entanglement horizons cannot increase its area. This monotonicity is the key ingredient in ruling out short-cuts in the graph model construction.

\section{Maximum Principles and Comparison for the Minimal Surface Operator}\label{sec:appendix_max_principle}

In this appendix we collect the versions of the weak and strong maximum principles
used in the main text, following the conventions of Leon Simon
(see e.g.\ \cite{SimonPDE}). We also record the corresponding comparison principles
for the minimal surface operator, derived using standard linearization arguments.

\subsection{Ellipticity and linear operators}

Let $\Omega\subset\mathbb{R}^n$ be a domain and consider a second--order linear operator
of the form
\begin{equation}\label{eq:linearL}
Lu := \sum_{i,j=1}^n a_{ij}(x) D_i D_j u
      + \sum_{j=1}^n b_j(x) D_j u
      + c(x) u .
\end{equation}

We assume:
\begin{itemize}
\item $a_{ij}, b_j, c \in L^\infty(\Omega)$,
\item $a_{ij}=a_{ji}$,
\item (\emph{uniform ellipticity}) there exists $\mu>0$ such that
\begin{equation}\label{eq:ellipticity}
\sum_{i,j=1}^n a_{ij}(x)\xi_i\xi_j \ge \mu |\xi|^2
\qquad \forall\, x\in\Omega,\ \xi\in\mathbb{R}^n .
\end{equation}
\end{itemize}

\subsection{Weak and strong maximum principles}

\begin{theorem}[Weak maximum principle]\label{thm:weakMP}
Let $\Omega$ be bounded and let $u\in C^2(\Omega)\cap C^0(\overline{\Omega})$
satisfy
\[
Lu \ge 0 \quad \text{in } \Omega ,
\]
where $L$ is as in \eqref{eq:linearL}--\eqref{eq:ellipticity} and
\begin{equation}\label{eq:cneg}
c(x)\le 0 \quad \text{in } \Omega .
\end{equation}
Then
\[
\max_{\overline{\Omega}} u \le \max_{\partial\Omega} u_+ ,
\qquad
u_+(x):=\max\{u(x),0\}.
\]
If $c\equiv 0$, then $\max_{\overline{\Omega}} u \le \max_{\partial\Omega} u$.
\end{theorem}

\begin{theorem}[Strong (Hopf) maximum principle]\label{thm:strongMP}
Let $\Omega$ be connected and let $u\in C^2(\Omega)$ satisfy
\[
Lu \ge 0 \quad \text{in } \Omega ,
\]
with $L$ as in \eqref{eq:linearL}--\eqref{eq:ellipticity} and $c\leq 0$ in $\Omega$.
If $u$ attains a nonnegative maximum at an interior point of $\Omega$, then
$u$ is constant in $\Omega$.
If $c\equiv 0$, the qualifier ``nonnegative'' may be dropped.
\end{theorem}


\subsection{The minimal surface operator}

Let
$$\mathcal{M}(u)
:= \operatorname{div}\!\left(\frac{\nabla u}{\sqrt{1+|\nabla u|^2}}\right)
= \sum_{i,j=1}^n
\left(\delta_{ij}-\frac{D_i u D_j u}{1+|Du|^2}\right)
\frac{D_{ij}u}{\sqrt{1+|Du|^2}} .$$

Writing $\mathcal{M}(u)=\sum_i D_i(A_i(Du))$ with
\[
A_i(p)=\frac{p_i}{\sqrt{1+|p|^2}},
\]
we note that
\[
\frac{\partial A_i}{\partial p_j}(p)
= \frac{\delta_{ij}}{\sqrt{1+|p|^2}}
  - \frac{p_i p_j}{(1+|p|^2)^{3/2}}
\]
is positive definite. Thus $\mathcal{M}$ is quasilinear elliptic and corresponds,
upon linearization, to an operator $L$ of the form \eqref{eq:linearL} with
$c\equiv 0$.

Geometrically, if $u$ is a graph over $\Omega\subset\mathbb{R}^n$ and
\[
\nu := \frac{(\nabla u,-1)}{\sqrt{1+|\nabla u|^2}}
\]
denotes the \emph{downward--pointing} unit normal to the graph, then
\[
\mathcal{M}(u)=\operatorname{div}(\nu).
\]
With this convention, $\mathcal{M}(u)=0$ is equivalent to vanishing mean curvature.
Moreover, if a smooth hypersurface is strictly mean--convex with respect to the
\emph{upward} normal, then when written locally as a graph with downward normal
it satisfies $\mathcal{M}(u)<0$. In particular, an upper hemisphere written as a graph
over its tangent plane at the north pole satisfies $\mathcal{M}(u)<0$, while the
tangent plane itself satisfies $\mathcal{M}(u)=0$.
\subsection{Weak comparison principle for $\mathcal{M}$}

\begin{proposition}[Weak comparison]\label{prop:weakcomp}
Let $\Omega$ be bounded and suppose
\[
\mathcal{M}(u_1)\leq \mathcal{M}(u_2) \quad \text{ in } \Omega ,
\]
with $\mathcal{M}(u_i)\in C^0(\overline{\Omega})$.
Assume
\[
u_1 = u_2 \quad \text{on } \partial\Omega .
\]
Then
\[
u_2 \le u_1 \quad \text{in } \Omega .
\]
\end{proposition}

\begin{proof}
Set $w:=u_2-u_1$. Writing $\mathcal{M}(u)=\sum_i D_i(A_i(Du))$,
\[
\mathcal{M}(u_2)-\mathcal{M}(u_1)
= \sum_{i,j} D_i\!\left(a_{ij}(x) D_j w\right),
\]
where
\[
a_{ij}(x)=\int_0^1
\frac{\partial A_i}{\partial p_j}\!\left(Du_1+tD(u_2-u_1)\right)\,dt .
\]
The matrix $(a_{ij})$ is uniformly elliptic.
Since $\mathcal{M}(u_2)-\mathcal{M}(u_1)\ge 0$, we obtain
\[
\sum_{i,j} a_{ij} D_iD_j w + \sum_i D_i(a_{ij})D_j w \ge 0 .
\]
Thus $w$ satisfies $Lw\ge 0$ with $c\equiv 0$ and $w=0$ on $\partial\Omega$.
The weak maximum principle yields $w\le 0$ in $\Omega$.
\end{proof}

\subsection{Strong comparison principle for $\mathcal{M}$}

\begin{proposition}[Strong comparison]\label{prop:strongcomp}
Let $\Omega$ be connected and suppose
\[
\mathcal{M}(u_1)\le \mathcal{M}(u_2) \quad \text{in } \Omega ,
\qquad
u_1 \ge u_2 \quad \text{in } \Omega .
\]
If there exists $x_0\in\Omega$ such that $u_1(x_0)=u_2(x_0)$, then
\[
u_1 \equiv u_2 \quad \text{in } \Omega .
\]
\end{proposition}

\begin{proof}
With $w:=u_2-u_1$ as above, we obtain $Lw\ge 0$ with $c\equiv 0$.
The hypotheses give $w\le 0$ in $\Omega$ and $w(x_0)=0$.
By the strong maximum principle (Theorem~\ref{thm:strongMP}),
$w$ is constant, hence $w\equiv 0$.
\end{proof}

\subsection{Geometric interpretation}
Proposition~\ref{prop:strongcomp} is equivalent to the geometric statement that
a minimal graph cannot touch, from the side on which it lies, another graph whose
mean curvature vector (computed with respect to the same choice of normal)
points strictly away from it, unless the two coincide locally. In particular, the
tangent plane at a point of a strictly mean--convex surface does not violate the
principle, since the required inequality on mean curvatures fails in that case.

\section{Geometric Interpretation of the Complete/Universal Graph}

We provide a geometric interpretation of the complete/universal graph with vertex set $\mathcal{V} = \{0,1\}^{\mathcal{P}(\nset)\setminus \emptyset}$, which was introduced algebraically in \cite{bao2015holographic} to represent the entropy cone. Each entanglement horizon $\HHRT(A_I)$ divides the bulk into two parts. A bitstring $x \in \mathcal{V}$ can be interpreted as indicating, for each composite region $I$, whether a given spacetime point is "inside" ($x_I=1$) or "outside" ($x_I=0$) the corresponding entanglement wedge. The vertex $x_i$ defined by $(x_i)_I = 1 \iff i \in I$ corresponds precisely to the bulk region causally connected to boundary $A_i$—that is, its entanglement wedge. Considering all $2^n-1$ possible composite regions and their associated (fully connected) extremal surfaces yields the finest possible partition of the bulk spacetime, naturally recovered by this complete set of bitstrings. This perspective elucidates the constructions in Lemma 6 and Proposition 7 of \cite{bao2015holographic}, framing them not as abstract combinatorial choices but as consequences of spacetime causality and wedge structure.

\bibliographystyle{JHEP}
\bibliography{biblio.bib}

@article{bao2015holographic,
  title={The holographic entropy cone},
  author={Bao, Ning and Nezami, Sepehr and Ooguri, Hirosi and Stoica, Bogdan and Sully, James and Walter, Michael},
  journal={Journal of High Energy Physics},
  volume={2015},
  number={9},
  pages={1--48},
  year={2015},
  publisher={Springer}
}

@article{grado2025minimax,
  title={Minimax surfaces and the holographic entropy cone},
  author={Grado-White, Brianna and Grimaldi, Guglielmo and Headrick, Matthew and Hubeny, Veronika E},
  journal={Journal of High Energy Physics},
  volume={2025},
  number={5},
  pages={1--58},
  year={2025},
  publisher={Springer}
}

@article{wall2014maximin,
  title={Maximin surfaces, and the strong subadditivity of the covariant holographic entanglement entropy},
  author={Wall, Aron C},
  journal={Classical and Quantum Gravity},
  volume={31},
  number={22},
  pages={225007},
  year={2014},
  publisher={IOP Publishing}
}

@article{zhao2026proof,
  title={A proof of the generalized Connected Wedge Theorem},
  author={Zhao, Bowen},
  journal={Journal of High Energy Physics},
  volume={2026},
  number={1},
  pages={29},
  year={2026},
  publisher={Springer}
}

@article{may2020holographic,
  title={Holographic scattering requires a connected entanglement wedge},
  author={May, Alex and Penington, Geoff and Sorce, Jonathan},
  journal={Journal of High Energy Physics},
  volume={2020},
  number={8},
  pages={1--34},
  year={2020},
  publisher={Springer}
}

@article{zhao2025beyond,
  title={Beyond $2 $-to-$2 $: Geometrization of Entanglement Wedge Connectivity in Holographic Scattering},
  author={Zhao, Bowen},
  journal={arXiv preprint arXiv:2512.06815},
  year={2025}
}

@article{RT2006formula,
  title={Holographic Derivation of Entanglement Entropy from the anti--de Sitter Space/Conformal Field Theory Correspondence},
  author={Ryu, Shinsei and Takayanagi, Tadashi},
  journal={Physical review letters},
  volume={96},
  number={18},
  pages={181602},
  year={2006},
  publisher={APS}
}

@article{HRT2007covariant,
  title={A covariant holographic entanglement entropy proposal},
  author={Hubeny, Veronika E and Rangamani, Mukund and Takayanagi, Tadashi},
  journal={Journal of High Energy Physics},
  volume={2007},
  number={07},
  pages={062},
  year={2007},
  publisher={IOP Publishing}
}

@unpublished{SimonPDE,
  author       = {Simon, Leon},
  title        = {Lectures on {PDE}},
  note         = {Unpublished {T}singhua lecture notes, Stanford University},
  year         = {2022},
}

@book{waldGR,
  title={General relativity},
  author={Wald, Robert M},
  year={2024},
  publisher={University of Chicago press}
}

@article{lima2025sufficientGCWT,
  title={On sufficient conditions for holographic scattering},
  author={Lima, Caroline and Pasterski, Sabrina and Waddell, Chris},
  journal={arXiv preprint arXiv:2509.26264},
  year={2025}
}

@article{witten2020light,
  title={Light rays, singularities, and all that},
  author={Witten, Edward},
  journal={Reviews of Modern Physics},
  volume={92},
  number={4},
  pages={045004},
  year={2020},
  publisher={APS}
}

@article{HH2023minimax,
  title={Covariant bit threads},
  author={Headrick, Matthew and Hubeny, Veronika E},
  journal={Journal of High Energy Physics},
  volume={2023},
  number={7},
  pages={180},
  year={2023},
  publisher={Springer}
}

@article{HHLR2014,
  title={Causality \& holographic entanglement entropy},
  author={Headrick, Matthew and Hubeny, Veronika E and Lawrence, Albion and Rangamani, Mukund},
  journal={Journal of High Energy Physics},
  volume={2014},
  number={12},
  pages={1--36},
  year={2014},
  publisher={Springer}
}

@article{bao2018entropy_large_region_late_time,
  title={On the entropy cone for large regions at late times},
  author={Bao, Ning and Mezei, M{\'a}rk},
  journal={arXiv preprint arXiv:1811.00019},
  year={2018}
}

@article{erdmenger2017HECnumerics,
  title={Time evolution of entanglement for holographic steady state formation},
  author={Erdmenger, Johanna and Fern{\'a}ndez, Daniel and Flory, Mario and Meg{\'\i}as, Eugenio and Straub, Ann-Kathrin and Witkowski, Piotr},
  journal={Journal of High Energy Physics},
  volume={2017},
  number={10},
  pages={1--58},
  year={2017},
  publisher={Springer}
}

@article{caginalp2020holographicAdSVaidya_numerics,
  title={Holographic entropy cone in AdS-Vaidya spacetimes},
  author={Caginalp, Reginald J},
  journal={Physical Review D},
  volume={101},
  number={2},
  pages={026010},
  year={2020},
  publisher={APS}
}

@article{czech2019holographic2+1dimension,
  title={Holographic entropy cone with time dependence in two dimensions},
  author={Czech, Bart-lomiej and Dong, Xi},
  journal={Journal of High Energy Physics},
  volume={2019},
  number={10},
  pages={177},
  year={2019},
  publisher={Springer}
}

@article{grimaldi2025newcharacterHEC,
  title={A new characterization of the holographic entropy cone},
  author={Grimaldi, Guglielmo and Headrick, Matthew and Hubeny, Veronika E},
  journal={arXiv preprint arXiv:2508.21823},
  year={2025}
}

@article{bao2019holographicTensorNetwork,
  title={Beyond toy models: distilling tensor networks in full AdS/CFT},
  author={Bao, Ning and Penington, Geoffrey and Sorce, Jonathan and Wall, Aron C},
  journal={Journal of High Energy Physics},
  volume={2019},
  number={11},
  pages={69},
  year={2019},
  publisher={Springer}
}


\end{document}